\documentclass[final,5p,times,twocolumn]{elsarticle}
\biboptions{compress}
\journal{Automatica}
\pdfoutput=1

\usepackage{amssymb}
\usepackage{amsmath,amsfonts}
\usepackage{algorithmic}
\usepackage{algorithm}
\usepackage{array}
\usepackage{textcomp}
\usepackage{stfloats}
\usepackage{url}
\usepackage{verbatim}
\usepackage{graphicx}

\usepackage{amsmath,amssymb,amsfonts,amsthm} 
\usepackage{xcolor}

\usepackage{appendix}

\usepackage{wrapfig} 

\allowdisplaybreaks

\newtheorem{theorem}{Theorem}
\newtheorem{lemma}[theorem]{Lemma}
\newdefinition{ass}{Assumption}
\newdefinition{rem}{Remark}
\newdefinition{defn}{Definition}
\newtheorem{pro}{Proposition}

\allowdisplaybreaks

\begin{document}	
	\begin{frontmatter}
		\title{Online Parameter Identification of Generalized Non-cooperative Game\tnoteref{label1}}
				
		\author[1,2]{Jianguo Chen}\ead{chenjianguo@amss.ac.cn}
		\author[3,4]{Jinlong Lei }\ead{leijinlong@tongji.edu.cn}
		\author[1,2]{Hongsheng Qi}\ead{qihongsh@amss.ac.cn}
		\author[3,4]{Yiguang Hong}\ead{yghong@iss.ac.cn}
		

\address[1]{Key Laboratory of Systems and Control, Academy of Mathematics and Systems Science, Chinese Academy of Sciences, Beijing 100190, China;}
\address[2]{School of Mathematical Sciences, University of Chinese Academy of Sciences, Beijing 100049, China;}
\address[3]{Department of Control Science and Engineering, Tongji University, Shanghai 201804, China;}
\address[4]{Shanghai Research Institute for Intelligent Autonomous Systems, Shanghai 201210, China}

\begin{abstract}
This work studies the parameter identification problem of a generalized non-cooperative game, where each player's cost function is influenced by an observable signal and some unknown parameters. We consider the scenario where equilibrium  of the game at some observable signals can be observed with  noises, whereas our goal is   to identify the unknown parameters with the observed data.  Assuming that the observable signals and the corresponding noise-corrupted equilibriums are acquired sequentially, we construct this parameter identification problem as  online optimization and introduce a novel online parameter identification algorithm. To be specific, we construct a regularized loss function that balances conservativeness and correctiveness, where the conservativeness term ensures that the new estimates do not deviate significantly from the current estimates, while the correctiveness term is captured by the Karush-Kuhn-Tucker conditions. We then prove that when the players' cost functions are linear with respect to the unknown parameters and the learning rate of  the online parameter identification algorithm  satisfies $\mu_k \propto 1/\sqrt{k}$, along with other  assumptions, the regret bound of   the proposed algorithm is $O(\sqrt{K})$. Finally, we conduct numerical simulations on a Nash-Cournot problem to demonstrate that the performance of the online identification algorithm is comparable to that of the offline setting.
\end{abstract}

\begin{keyword}
Parameter identification\sep online learning\sep generalized non-cooperative game\sep inverse game\sep regret bound.
\end{keyword}
\end{frontmatter}

\section{Introduction}
The class of non-cooperative games  has found  wide applications in various domains, including  network security, urban traffic management, and power systems \cite[]{Roy,Alvarez,mei}. In cases where the feasible set of a player is influenced by the actions of other players, it is known as a generalized non-cooperative game (GNCG) \cite[]{genp1,gnep2,gnep3,gnep4,gnep5}. The model is important in the realm of economic sciences and has found utility in diverse fields like electricity \cite[]{gnep4,gnep7} and natural gas markets \cite[]{gnep8}. 
The fundamental concept in GNCG is the generalized Nash equilibrium (GNE), which is crucial for predicting individual players' strategies.

The computation of GNE heavily relies on the access to the cost functions of all players \cite[]{fran}. However, in practical settings, it is often the case that we can only observe the equilibrium behaviors of players in a given game while remaining unaware of the specific cost functions underlying the game \cite[]{fangfei}. For instance, in a competitive market consisting of multiple companies, we might be able to observe market pricing and product  volumes but lack precise information about production costs. Despite the parametric  uncertainty  in cost functions, we may still be able to identify these parameters through the observed equilibriums. Since knowing the parameterized game model  enables  us to make predictions about  players' future behaviors, parameter identification in the context of GNCG has significant importance in various fields such as autonomous driving \cite[]{Le}.

Real-life scenarios often present challenges where direct access to all equilibria is impossible. Instead, we are faced with the task of constantly observing new equilibrium results as the environment  evolves \cite[]{online2}. For instance, in a competitive market, the game is perpetually in progress, with companies continuously reaching new equilibria as the external market conditions change. In such dynamic situations, online parameter identification becomes necessary. Thus, this article focuses on the online parameter identification of GNCG.

\subsection{Literature Review}
The field of systems and control has witnessed extensive research on the parameter identification problem, particularly in the context of linear \cite[]{linear1,linear2,linear3} and nonlinear systems \cite[]{nonlinear1,nonlinear2,nonlinear3}. However, there remains a dearth of investigation regarding parameter identification in game systems, with a particular emphasis on online identification. Some literatures refer to parameter identification in non-cooperative games as inverse games. To address the challenges associated with parameter identification in differential game problems, researchers have employed the Pontryagin's maximum principle derived from optimal control \cite[]{Molloy, Cao, molloy2019, rothfu2017}. Moreover, an inverse optimization method has been specifically developed to facilitate the estimation of parameters within the cost functions of traffic flow games, which are modeled as generalized Nash games \cite[]{Allen}. When it comes to matrix games, the inverse problem primarily involves estimating the cost matrix \cite[]{fangfei, yu2022}. Furthermore, inverse reinforcement learning has been utilized to explore the intricacies of inverse Markov games \cite[]{lin2017}. However, it is worth noting that most of the aforementioned works on inverse games assume the availability of simultaneously observed data, resulting in the proposal of offline algorithms.

In \cite{Le}, an online parameter identification algorithm has been proposed for a game utilized in autonomous driving, employing the unscented Kalman filter. Similarly, \cite{Zhang} has introduced an online inverse dynamic game algorithm for linear quadratic games, inspired by online inverse optimal control algorithms, along with an analysis of solution uniqueness. However, the performance of these online algorithms has not been thoroughly examined. To the best of our knowledge, a comprehensive investigation of online parameter identification for games in a general setting is still lacking.

\subsection{Contributions}

This work focuses on the online parameter identification of GNCG. The main contributions can be summarized as follows:
\begin{itemize}
  \item We model the problem as an online convex optimization problem and propose an online parameter identification algorithm for GNCG. Specially, we  delicately  design a regularized loss function to   balance between conservativeness and correctiveness.
  \item We prove that when the players' cost functions are linear in the unknown parameters, the learning rate of the algorithm $\mu_k \propto 1/\sqrt{k}$ and other assumptions are satisfied, the regret bound of the online parameter identification algorithm is $O(\sqrt{K})$.
  \item Through simulations on a Nash-Cournot problem, we demonstrate that the performance of the online parameter identification algorithm closely resembles that of the algorithm in the whole batch setting after a few rounds.
\end{itemize}

\subsection{Paper Organization and Notation}
\emph{Paper Organization}: Section \ref{s2}  formulates the online parameter identification of GNCG. In Section \ref{s3}, we propose an online  algorithm for identifying    parameters of GNCG, while the regret bound is established in Section \ref{s4} with some proofs given in   Appendices. Section \ref{s5}  showcases numerical simulations on a natural gas market, and some concluding remarks are provided in Section \ref{s6}. 

\emph{Notations}: $\mathcal{R}^n$ represents the n-dimensional vector space; $\mathcal{R}^{n\times m}$ represents the space of $n \times m$-dimensional matrices; $\|\cdot\|_2$ denotes the $\mathcal{L}_2$-norm; bold letters denote vectors or matrices; $\mathrm{col}(\boldsymbol{x}_1,\cdots,\boldsymbol{x}_N)$ represents the column vector $(\boldsymbol{x}_1^T,\cdots,\boldsymbol{x}_N^T)^T$; vector $\boldsymbol{v}>0$ indicates that every element is positive, and vector $\boldsymbol{v}\geq0$ indicates that every element is non-negative; $\boldsymbol{a} \perp \boldsymbol{b}$ signifies that the product of the corresponding elements of vectors $\boldsymbol{a}$ and $\boldsymbol{b}$ is equal to 0; $\langle \cdot , \cdot \rangle$ denotes the inner product between vectors; and $R_{\geq0}$ denotes the set of non-negative real numbers.

\section{Problem Statement} \label{s2}

In this section, we present the problem statement concerning the online parameter identification of generalized non-cooperative games.

\subsection{Generalized Non-cooperative Game}
The generalized non-cooperative game (GNCG) is composed of $N$ players, denoted as $\mathcal{N} := \{1,\cdots,N\}$. Each player $v \in \mathcal{N}$ possesses control over its strategy $\boldsymbol{x}_v \in \boldsymbol{X}_v \subseteq \mathcal{R}^{n_v}$, where $\boldsymbol{X}_v$ represents the feasible set of decision variables for player $v$. For every player $v \in \mathcal{N}$, considering the parameter $\boldsymbol{\theta}_v \in \boldsymbol{\Theta}_v \subseteq \mathcal{R}^{n_v'}$, the observable signal $\boldsymbol{u}_v \in \boldsymbol{U}_v \subseteq \mathcal{R}^{n_v''}$ and the decision variables of other players $\boldsymbol{x}_{-v} \in \boldsymbol{X}_{-v} := \times_{s\neq v}\boldsymbol{X}_s$, it selects a strategy $\boldsymbol{x}_v$ to minimize the optimization problem
\begin{equation}\label{q}
  \begin{aligned}
&\min_{\boldsymbol{x}_v} f_v(\boldsymbol{x}_v;\boldsymbol{x}_{-v},\boldsymbol{u}_v,\boldsymbol{\theta}_v) \\
&s.t. \ \boldsymbol{x}_v \in \boldsymbol{X}_v(\boldsymbol{x}_{-v},\boldsymbol{u}_v),
  \end{aligned}
\end{equation}
\noindent where $f_v:\mathcal{R}^{n \times n_v'' \times n_v'} \mapsto \mathcal{R}$, with $n:=\sum_{v=1}^{N}n_v$, represents the cost function of player $v$. The feasible set $\boldsymbol{X}_v$ depends on the decision variables of other players and the signal. In the following, we impose the assumption that each player's optimization problem is convex.
\begin{ass} \label{convex}
	For every player $v \in \mathcal{N}$, $\boldsymbol{u}_v \in \boldsymbol{U}_v$, $\boldsymbol{\theta}_v \in \boldsymbol{\Theta}_v$ and $\boldsymbol{x}_{-v} \in \boldsymbol{X}_{-v}$, the cost function $f_v(\cdot;\boldsymbol{x}_{-v},\boldsymbol{u}_v,\boldsymbol{\theta}_v)$ is convex and continuously differentiable in $\boldsymbol{X}_v(\boldsymbol{x}_{-v},\boldsymbol{u}_v)$. Moreover, the sets $X_v(\boldsymbol{x}_{-v},\boldsymbol{u}_v)$ and $\boldsymbol{\Theta}_v$ are closed and convex.
\end{ass}

For the sake of simplicity,   variables pertaining to all players are stacked and denoted as $\boldsymbol{x}=\mathrm{col}(\boldsymbol{x}_1,\cdots,\boldsymbol{x}_N)$, $\boldsymbol{u}=\mathrm{col}(\boldsymbol{u}_1,\cdots,\boldsymbol{u}_N) \in \boldsymbol{U}$ and $\boldsymbol{\theta} = \mathrm{col}(\boldsymbol{\theta}_1,\cdots,\boldsymbol{\theta}_N) \in \boldsymbol{\Theta}$. Here, $\boldsymbol{\Theta}$ is a closed and convex set from Assumption \ref{convex}. The generalized Nash equilibrium stands as a crucial solution concept within the GNCG framework, and its formal definition is as follows.
\begin{defn} \label{gne}
Given $\boldsymbol{\theta} \in \boldsymbol{\Theta}$ and $\boldsymbol{u} \in \boldsymbol{U}$, let $\boldsymbol{x}^* = \mathrm{col}(\boldsymbol{x}_1^*, \cdots,\boldsymbol{x}_N^*)$ be a generalized Nash equilibrium (GNE) if, for every player $v \in \mathcal{N}$, the following condition holds:
\begin{equation*}
  \begin{aligned}
    f_v(\boldsymbol{x}_v^*;\boldsymbol{x}_{-v}^*,\boldsymbol{u}_v,\boldsymbol{\theta}_v) \leq f_v(\boldsymbol{x}_v;\boldsymbol{x}_{-v}^*,\boldsymbol{u}_v,\boldsymbol{\theta}_v),\ \forall \boldsymbol{x}_{v} \in X_v(\boldsymbol{x}_{-v}^*,\boldsymbol{u}_v).
  \end{aligned}
\end{equation*}
\end{defn}

GNE refers to a scenario in the GNCG where each player's current strategy minimizes  its cost function, while the strategies of other players remain unchanged. According to Theorem 4.1 in \cite{fran}, the existence of a GNE, as defined by Definition \ref{gne}, is guaranteed under Assumption \ref{convex}.

One extensively studied coupling constraint set within the GNCG framework is the jointly convex GNCG \cite[]{gnep6} shown below.
\begin{ass} \label{jointly}
The GNCG is jointly convex in this article, i.e.,  there exists a closed convex set $\boldsymbol{X}(\boldsymbol{u}) \subseteq \mathcal{R}^{n}$ associated with the signal $\boldsymbol{u}$ such that for each player $v$,
\[X_v(\boldsymbol{x}_{-v},\boldsymbol{u}_v) = \{ \boldsymbol{x}_v | (\boldsymbol{x}_v,\boldsymbol{x}_{-v}) \in \boldsymbol{X}(\boldsymbol{u}) \}.\]
\noindent Additionally, $X(\boldsymbol{u})$ is given by $\{ \boldsymbol{x} \ | \ h_q(\boldsymbol{x},\boldsymbol{u}) \le 0, 1 \leq q \leq m ; g_j(\boldsymbol{x},\boldsymbol{u}) = 0, 1 \leq j \leq p \}$, where $m$ and $n$ are the number of inequality and equality constraints, respectively. Moreover, let $h_q(\cdot,\boldsymbol{u})$ and $g_j(\cdot,\boldsymbol{u})$ be continuously differentiable functions.
\end{ass}

\subsection{Online Parameter Identification of GNCG based on Noisy Equilibrium  Observation}
The computation of equilibrium in  games typically necessitates knowledge of the cost functions for all players \cite[]{fran}. However, in real-life scenarios, we may not have access to the players' cost functions directly. Instead, we can observe the equilibrium solutions of these problems \cite[]{fangfei}. Let $\boldsymbol{y}$ represent the observed equilibrium of a game problem \eqref{q}, for which the cost functions are parameterized by an unknown parameter $\boldsymbol{\theta}$. Our objective is to estimate the unknown parameter $\boldsymbol{\theta}$   based on the observed equilibrium $\boldsymbol{y}$. This process is known as parameter identification in the game. In particular, we consider a scenario where the observed equilibrium is subject to noise, i.e., $\boldsymbol{y}=\boldsymbol{x} + \boldsymbol{\epsilon}$, where $\boldsymbol{x}$ denotes the true equilibrium, while $\boldsymbol{\epsilon}$ represents a random variable following a specific distribution.

 We focus on the scenario where data is sequentially observed, namely, in the $k$-th round, we have access to a signal $\boldsymbol{u}^k$ and a noise-corrupted equilibrium $\boldsymbol{y}^k$. Let $\boldsymbol{\theta}^1$ denote an initial estimate of the unknown parameters, and $\boldsymbol{\theta}^{k+1} \in \Theta$ represent  the estimate in the $k$-th round. Upon obtaining new observed data $(\boldsymbol{y}^k, \boldsymbol{u}^k)$ in the $k$-th round, we can update the estimate $\boldsymbol{\theta}^{k+1} \in \Theta$ using an online learning algorithm guided by a well-designed loss function $l(\boldsymbol{\theta};\boldsymbol{y}^k, \boldsymbol{u}^k)$. The performance evaluation of such an online parameter identification algorithm of GNCG is assessed by the regret defined as follows.
\begin{defn} \label{regretdefi}
(Regret of online parameter identification.)
\begin{equation} \label{regret}
\begin{aligned}
R_K := \sum_{k=1}^{K} l(\boldsymbol{\theta}^k;\boldsymbol{y}^k, \boldsymbol{u}^k) - \sum_{k=1}^{K} l(\boldsymbol{\theta}_*^K;\boldsymbol{y}^k, \boldsymbol{u}^k),
\end{aligned}
\end{equation}
\noindent where $\boldsymbol{\theta}_*^K$ represents the optimal inference within $\boldsymbol{\Theta}$ that minimizes $\sum_{k=1}^{K}l(\boldsymbol{\theta};\boldsymbol{y}^k,\boldsymbol{u}^k)$ in the whole batch setting.
\end{defn}

The regret can reflect the performance of an online parameter identification algorithm by comparing the cumulative loss function with that in the whole batch learning. An online learning algorithm is said to have the no-regret property if $R_K = o(K)$.
\section{Online Parameter Identification Algorithm of GNCG } \label{s3}
In this section, we first define the loss function based on Karush-Kuhn-Tucker (KKT) conditions of the GNE. Subsequently, we design an online parameter identification.
\subsection{Loss Function based on KKT Conditions}
Typically, a GNCG can have multiple GNEs. In practice, a specific class of economically meaningful GNEs, called variational equilibriums, has been formulated. These equilibriums represent a refined subset of the GNE \cite[]{Kulkarni}. Therefore, we assume that the observed equilibrium results belong to the class of variational equilibriums.
\begin{defn}\label{ve}
Let Assumptions \ref{convex}-2 hold. $\boldsymbol{x}^*$ qualifies as a variational equilibrium if, given $\boldsymbol{u} \in \boldsymbol{U}$ and $\boldsymbol{\theta} \in \boldsymbol{\Theta}$, it satisfies the variational inequality $VI(\boldsymbol{X}, \boldsymbol{F})$, i.e.,
\begin{align}\label{vi}
\boldsymbol{F}(\boldsymbol{x}^*,\boldsymbol{u},\boldsymbol{\theta})^T (\boldsymbol{y} - \boldsymbol{x}^*) \geq 0, \ \forall \boldsymbol{y} \in X(\boldsymbol{u}),
\end{align}
\noindent where $\boldsymbol{F}(\boldsymbol{x},\boldsymbol{u},\boldsymbol{\theta}) := (\nabla_{\boldsymbol{x}_v}f_v(\boldsymbol{x}_v;\boldsymbol{x}_{-v},\boldsymbol{u}_v,\boldsymbol{\theta}_v))_{v=1}^N$.
\end{defn}

With $X(\boldsymbol{u})$ defined in Assumptions  2, the KKT conditions corresponding to variational inequality (\ref{vi}) are   as follows.
\begin{equation}\label{kkt2}
\begin{cases}
\boldsymbol{F}(\boldsymbol{x},\boldsymbol{u},\boldsymbol{\theta}) +  \sum_{q=1}^m \lambda_q \nabla_{\boldsymbol{x}} h_q(\boldsymbol{x},\boldsymbol{u}) + \sum_{j=1}^p \nu_j \nabla_{\boldsymbol{x}} g_j(\boldsymbol{x},\boldsymbol{u}) = \boldsymbol{0} \\
0 \le \lambda_q \perp h_q(\boldsymbol{x},\boldsymbol{u}) \le 0 ,\ \forall q, 1 \leq q \leq m \\
g_j(\boldsymbol{x},\boldsymbol{u}) = 0, \ \forall j, 1 \leq j \leq p,
\end{cases}
\end{equation}
where $\lambda_q, q=1,\cdots,m$ and $\nu_j, j=1,\cdots,p$ are the dual variables corresponding to the inequality and equality constraints, respectively.
Denote by $\boldsymbol{\lambda}=\mathrm{col}(\lambda_1,\cdots, \lambda_m)$ and $\boldsymbol{\nu}=\mathrm{col}(\nu_1,\cdots,\nu_p) .$
Subsequently, we define the loss function $l\left(\boldsymbol{\theta};\boldsymbol{y}, \boldsymbol{u}\right)$ based on KKT system (\ref{kkt2}) as follows.

\begin{defn}\label{loss-defi2}
(Loss function.) For a given signal $\boldsymbol{u} \in \boldsymbol{U}$ and the corresponding observed equilibrium $\boldsymbol{y}$, set $l\left(\boldsymbol{\theta};\boldsymbol{y}, \boldsymbol{u}\right):=\min_{\boldsymbol{\lambda}, \boldsymbol{\nu}} \left\{L\left(\boldsymbol{\theta}, \boldsymbol{\lambda}, \boldsymbol{\nu};\boldsymbol{y}, \boldsymbol{u}\right) \ | \ \boldsymbol{\lambda} \geq 0 \right\},$
where
\begin{align*}
 L&\left(\boldsymbol{\theta},\boldsymbol{\lambda}, \boldsymbol{\nu};\boldsymbol{y}, \boldsymbol{u}\right)=   \sum_{q=1}^{m} \chi \left(\lambda_q h_q(\boldsymbol{y},\boldsymbol{u})\right) + \sum_{j=1}^{p} \chi \left(g_j(\boldsymbol{y},\boldsymbol{u})\right)\\
 &+ \chi \biggl( \boldsymbol{F}(\boldsymbol{y},\boldsymbol{u},\boldsymbol{\theta}) +  \sum_{q=1}^m \lambda_q \nabla_{\boldsymbol{y}} h_q(\boldsymbol{y},\boldsymbol{u}) + \sum_{j=1}^p \nu_j \nabla_{\boldsymbol{y}} g_j(\boldsymbol{y},\boldsymbol{u})\biggr),
\end{align*}
\noindent and $\chi(\cdot):\mathcal{R}^n \mapsto R_{\geq0}$ is a penalty function with $\chi(\boldsymbol{0})=0$.
\end{defn}

\subsection{Online Parameter Identification Algorithm}
In the $k$-th round, we obtain a new observation $(\boldsymbol{y}^k, \boldsymbol{u}^k)$ and possess the estimate $\boldsymbol{\theta}^{k}$ from the previous $(k-1)$-round. We update the parameter using the learning method proposed in \cite{Kulis2010}, which strikes a balance between conservativeness and correctiveness. This implies that the new estimate $\boldsymbol{\theta}^{k+1}$ should align with the new observation while preserving some consistency with the previous estimate $\boldsymbol{\theta}^{k}$. As a result, the regularized loss function $G_k(\boldsymbol{\theta})$ can be defined as follows.
\begin{align} \label{update1}
	G_k(\boldsymbol{\theta}) = \underbrace{D\left(\boldsymbol{\theta},\boldsymbol{\theta}^k\right)}_{\mathrm{conservativeness}} + \mu _k \underbrace{l\left(\boldsymbol{\theta};\boldsymbol{y}^k, \boldsymbol{u}^k\right)}_{\mathrm{correctiveness}},
\end{align}
\noindent where $D(\cdot,\cdot): \mathcal{R}^{n'\times n'} \mapsto R_{\geq0}$ ($n'=\sum_{v=1}^{N}n_v'$) denotes a distance function, $\mu_k$ is the learning rate, and $l(\boldsymbol{\theta};\boldsymbol{y}^k, \boldsymbol{u}^k)$ is defined in Definition \ref{loss-defi2}. The first term of \eqref{update1} captures the `conservativeness' by evaluating the distance between $\boldsymbol{\theta}$ and $\boldsymbol{\theta}^k$, while the second term measures the `correctiveness' by assessing the concordance between $\boldsymbol{\theta}$ and the new observation $(\boldsymbol{y}^k, \boldsymbol{u}^k)$. The learning rate $\mu_k$ effectively balances these two aspects.

The parameter update involves two steps. Firstly, we get the optimal inference $\tilde{\boldsymbol{\theta}}^{k+1}$ by minimizing the regularized loss function. Secondly, we find a point $\boldsymbol{\theta}^{k+1}$ within $\boldsymbol{\Theta}$ that has the minimum distance from $\tilde{\boldsymbol{\theta}}^{k+1}$. The two steps are as follows:
\begin{subequations}
\begin{align} \label{update_phi}
\tilde{\boldsymbol{\theta}}^{k+1} &= \arg \min _{\boldsymbol{\theta}} G_k(\boldsymbol{\theta}), \\
\boldsymbol{\theta}^{k+1} &= \arg \min_{\boldsymbol{\theta} \in \boldsymbol{\Theta}} D\left(\boldsymbol{\theta} , \tilde{\boldsymbol{\theta}}^{k+1}\right).
\end{align}
\end{subequations}
The key computational cost of the update process lies in solving the optimization problem (\ref{update_phi}). As a matter of fact, with the definition $l\left(\boldsymbol{\theta};\boldsymbol{y} , \boldsymbol{u} \right)$ of in Definition \ref{loss-defi2}, we need to solve the following problem.
\begin{equation}
\begin{aligned}
&\min_{\boldsymbol{\theta}} \  D\left(\boldsymbol{\theta},\boldsymbol{\theta}^k\right) + \mu _k l\left(\boldsymbol{\theta};\boldsymbol{y}^k, \boldsymbol{u}^k\right) \\
=&\min_{\boldsymbol{\theta},\boldsymbol{\lambda}\geq0,\boldsymbol{\nu}} \  D\left(\boldsymbol{\theta},\boldsymbol{\theta}^k\right) + \mu _k L\left(\boldsymbol{\theta},\boldsymbol{\lambda}, \boldsymbol{\nu};\boldsymbol{y}^k, \boldsymbol{u}^k\right).
\end{aligned}
\end{equation}
As such, the parameter update step (\ref{update_phi}) is converted as
\begin{equation}
\begin{aligned} \label{update}
\left(\tilde{\boldsymbol{\theta}}^{k+1} , \boldsymbol{\lambda}^{k}, \boldsymbol{\nu}^{k}\right) = \arg \min_{\boldsymbol{\theta},\boldsymbol{\lambda}\geq0,\boldsymbol{\nu}} D\left(\boldsymbol{\theta},\boldsymbol{\theta}^k\right) + \mu _k L\left(\boldsymbol{\theta},\boldsymbol{\lambda}, \boldsymbol{\nu};\boldsymbol{y}^k, \boldsymbol{u}^k\right).
\end{aligned}
\end{equation}
Therefore, we summarize the procedures of the online parameter identification of GNCG   in Algorithm \ref{alg1}.
\begin{algorithm}\caption{Online parameter identification algorithm of GNCG}\label{alg1}
	\textbf{Require:} $\theta^1 \in \Theta$ \\
	1: Let $k \leftarrow 1$  \\
	2: \textbf{While} $k \leq K$ \textbf{do}  \\
	3: \quad  Observe $(\boldsymbol{y}^k, \boldsymbol{u}^k)$ \\
	4: \quad  Solve $\min_{\boldsymbol{\theta},\boldsymbol{\lambda}\geq0,\boldsymbol{\nu}} \  D\left(\boldsymbol{\theta},\boldsymbol{\theta}^k\right) + \mu _k L\left(\boldsymbol{\theta},\boldsymbol{\lambda}, \boldsymbol{\nu};\boldsymbol{y}^k, \boldsymbol{u}^k\right)$ to get $\tilde{\boldsymbol{\theta}}^{k+1}$ \\
	5: \quad  Solve $\min_{\boldsymbol{\theta} \in \boldsymbol{\Theta}} D\left(\boldsymbol{\theta} , \tilde{\boldsymbol{\theta}}^{k+1}\right)$ to get $\boldsymbol{\theta}^{k+1}$ \\
	6: \quad $k \leftarrow k+1$  \\
	7:\textbf{end}  	
\end{algorithm}

Generally speaking, the problem \eqref{update} might be nonconvex and cannot be exactly solved. While in this work, we will impose some suitable conditions on the problem, like the penalty and distance functions are convex, and \(\boldsymbol{\theta}\) is linearly structured in the cost functions, and which are shown in details in the next section. Because the addition and composition with an affine function both preserve convexity \cite{calafiore2014optimization}, the problem \eqref{update} is convex under those conditions. Therefore, step 4 can be solved by calling the well-known iterative optimization algorithms, like Trust-Region Constrained Algorithm and Sequential Quadratic Programming \cite{nocedal1999numerical}.

\section{Theoretical Analysis} \label{s4}
In this section, we will show that under certain conditions, such as taking  penalty and distance functions as the square of the $\mathcal{L}_2$-norm,   the cost function of each player $v\in \mathcal{N}$ is linear in the unknown parameter $\boldsymbol{\theta}_v$, and setting the learning rate as $\mu_k \propto 1/\sqrt{k}$, the regret bound produced by Algorithm \ref{alg1} is $O(\sqrt{K})$.
\subsection{Assumptions and Preliminary Results}
For the simplicity, we define
\begin{equation}
\begin{aligned} \label{notation}
&\boldsymbol{F}_{\boldsymbol{\theta}}:=\boldsymbol{F}\left(\boldsymbol{y}, \boldsymbol{u}, \boldsymbol{\theta}\right) \in \mathcal{R}^{n},\\
&\nabla \boldsymbol{h}:=\left[\nabla_{\boldsymbol{y}} h_1(\boldsymbol{y},\boldsymbol{u}),\cdots,\nabla_{\boldsymbol{y}} h_m(\boldsymbol{y},\boldsymbol{u})\right] \in \mathcal{R}^{n \times m} ,\\
&\nabla \boldsymbol{g}:=\left[\nabla_{\boldsymbol{y}} g_1(\boldsymbol{y},\boldsymbol{u}),\cdots,\nabla_{\boldsymbol{y}} g_p(\boldsymbol{y},\boldsymbol{u})\right] \in \mathcal{R}^{n \times p}, \\
&\boldsymbol{H}:=\mathrm{diag}(h_1(\boldsymbol{y},\boldsymbol{u}),\cdots,h_m(\boldsymbol{y},\boldsymbol{u}))\in \mathcal{R}^{m\times m} ,\\
&\boldsymbol{g}:=\left[g_1(\boldsymbol{y},\boldsymbol{u}),\cdots,g_p(\boldsymbol{y},\boldsymbol{u})\right]^T \in \mathcal{R}^{p} ,\\
&\boldsymbol{\lambda} := (\lambda_1,\cdots,\lambda_m)^T \in \mathcal{R}^{m}, \boldsymbol{\nu} := (\nu_1,\cdots,\nu_p)^T \in \mathcal{R}^{p}.
\end{aligned}
\end{equation}

Next, we present an assumption regarding the linearity of the players' cost functions in the game with respect to the unknown parameters $\boldsymbol{\theta}$. While this assumption may appear stringent, it can be justified in scenarios involving linear weights of known basis functions, as exemplified by the cost function of game-theoretic autonomous vehicles in \cite{li2017}.
\begin{ass}\label{assu-1}
For every $v \in \mathcal{N}$, $f_v(\boldsymbol{x}_v;\boldsymbol{x}_{-v},\boldsymbol{u}_v,\boldsymbol{\theta}_v)$ is linear with respect to $\boldsymbol{\theta}_v$, i.e., there exists basic functions $\tilde{\boldsymbol{f}_v}(\boldsymbol{x}_v;\boldsymbol{x}_{-v},\boldsymbol{u}_v):\mathcal{R}^{n \times n_v''}\mapsto \mathcal{R}^{n_v'}$ and $\breve{f_v}(\boldsymbol{x}_v;\boldsymbol{x}_{-v},\boldsymbol{u}_v):\mathcal{R}^{n \times n_v''}\mapsto \mathcal{R}$, which do not depend on the parameter $\boldsymbol{\theta}$ (Abbreviated as $f_v$, $\tilde{\boldsymbol{f}_v}$, $\breve{f_v}$ in the following), such that $f_v = \boldsymbol{\theta}_v^T \tilde{\boldsymbol{f}_v} + \breve{f_v}$.
\end{ass}

Next, we introduce an assumption regarding the boundedness of certain variables.
\begin{ass} \label{assu-2}
(1) $\boldsymbol{y}$, $\boldsymbol{u}$ and $\boldsymbol{\theta}$  are bounded, i.e., there exists a real number $B_1 > 0$, such that for every $\boldsymbol{u} \in \boldsymbol{U}, \boldsymbol{\theta} \in  \boldsymbol{\Theta}$ and $\boldsymbol{y}$, we have $\left\|\boldsymbol{u}\right\|_2$, $\left\|\boldsymbol{y}\right\|_2$ and $\left\|\boldsymbol{\theta}\right\|_2 <B_1$. (2) Furthermore, gradients of the basic cost functions, constraint functions, and their gradients are uniformly bounded, i.e., there exists a real number $B_2>0$, such that for every $v \in \mathcal{N}$, $ \boldsymbol{u} \in \boldsymbol{U} , \boldsymbol{y}$ and $\boldsymbol{\theta} \in \boldsymbol{\Theta}$, we have $\left\|\nabla_{\boldsymbol{x}_v}\tilde{\boldsymbol{f}_v}\right\|_2$, $\left\|\nabla_{\boldsymbol{x}_v}\breve{f_v}\right\|$, $\left\| \boldsymbol{H}\right\|_2$, $\left\|\nabla \boldsymbol{h} \right\|_2$ and $\left\|\nabla \boldsymbol{g}\right\|_2<B_2$.
\end{ass}

Building upon Assumptions \ref{assu-1} and \ref{assu-2}, we can establish the boundedness and Lipschitz continuity of $\boldsymbol{F}_{\boldsymbol{\theta}}$ defined by \eqref{notation}.
\begin{pro} \label{coro-1}
	Let Assumptions \ref{assu-1} and \ref{assu-2} hold. Then there exists constants $M_1, M_2>0$, such that for every $\boldsymbol{u} \in \boldsymbol{U}$ and $\boldsymbol{y}$,
	\begin{align} \label{bound_F}
	&\left\| \boldsymbol{F}_{\boldsymbol{\theta}} \right\|_2 \leq M_1, \forall \boldsymbol{\theta} \in  \boldsymbol{\Theta}, \\ \label{lip_F}
	&\left\| \boldsymbol{F}_{\boldsymbol{\theta}_1} - \boldsymbol{F}_{\boldsymbol{\theta}_2} \right\|_2 \leq M_2 \left\|\boldsymbol{\theta}_1 - \boldsymbol{\theta}_2 \right\|_2, \forall \boldsymbol{\theta}_1, \boldsymbol{\theta}_2 \in  \boldsymbol{\Theta}.
	\end{align}
\end{pro}
\begin{proof}
From the linearity of $f_v$ with respect to $\boldsymbol{\theta}_v$, we can obtain the structure of $\boldsymbol{F}_{\boldsymbol{\theta}}$ as follows.
\begin{equation} \label{F_theta}
	\begin{aligned}
	\boldsymbol{F}_{\theta}=\begin{bmatrix}
				\nabla_{\boldsymbol{x}_1} \tilde{\boldsymbol{f}_1} & & \\
				& \ddots & \\
				& & \nabla_{\boldsymbol{x}_N} \tilde{\boldsymbol{f}_N}
		  \end{bmatrix}^T\boldsymbol{\theta} + \begin{bmatrix}
			  \nabla_{\boldsymbol{x}_1} \breve{f_1} \\
			  \vdots \\
			  \nabla_{\boldsymbol{x}_N} \breve{f_N}
		  \end{bmatrix},
	\end{aligned}
\end{equation}
\noindent where $\boldsymbol{\theta} = \mathrm{col}(\boldsymbol{\theta_1},\cdots,\boldsymbol{\theta_N})$.
Because of the boundedness of $\|\boldsymbol{\theta}\|_2$, $\left\|\nabla_{\boldsymbol{x}_v}\tilde{\boldsymbol{f}_v}\right\|_2$ and $\left\|\nabla_{\boldsymbol{x}_v}\breve{f_v}\right\|_2, v \in \mathcal{N}$ from Assumption \ref{assu-2}, according to the structure of $\boldsymbol{F}_{\boldsymbol{\theta}}$ in (\ref{F_theta}), there exists a positive constant $M_1$ such that  (\ref{bound_F}) holds.

From (\ref{F_theta}), we have that for every $\boldsymbol{\theta}_1, \boldsymbol{\theta}_2 \in  \boldsymbol{\Theta}$,
\[\boldsymbol{F}_{\boldsymbol{\theta}_1} - \boldsymbol{F}_{\boldsymbol{\theta}_2} =  \begin{bmatrix}
	\nabla_{\boldsymbol{x}_1}\tilde{\boldsymbol{f}_1} & & \\
	  & \ddots & \\
	  & & \nabla_{\boldsymbol{x}_N} \tilde{\boldsymbol{f}_N}
\end{bmatrix}^T(\boldsymbol{\theta_1}-\boldsymbol{\theta_2}).\]
Since $\left\|\nabla _{\boldsymbol{x}_v} \tilde{\boldsymbol{f}_v}\right\|_2<B_2, v \in \mathcal{N}$ from Assumption \ref{assu-2}, (\ref{lip_F}) holds for a positive constant $M_2$.
\end{proof}

The following assumption provides concrete selections for the distance function and the loss function.
\begin{ass} \label{penalty}
 (1) The distance function is $D(\boldsymbol{x},\boldsymbol{y})=\frac{1}{2}\|\boldsymbol{x}-\boldsymbol{y}\|_2^2$;
	(2) The penalty function is the square of the $\mathcal{L}_2$-norm. Thus, the loss function defined in Definition \ref{loss-defi2}  is given by
		\begin{equation} \label{newll}
			\begin{aligned}
			l\left(\boldsymbol{\theta};\boldsymbol{y}, \boldsymbol{u}\right) = \min_{\boldsymbol{\lambda} \geq 0, \boldsymbol{\nu}} \{ \left\|\boldsymbol{F}_{\boldsymbol{\theta}} + \nabla \boldsymbol{h}\boldsymbol{\lambda} + \nabla \boldsymbol{g} \boldsymbol{\nu}\right\|_2^2 + \left\|\boldsymbol{H}\boldsymbol{\lambda}\right\|_2^2 + \|\boldsymbol{g}\|_2^2 \}.
			\end{aligned}
		\end{equation}
\end{ass}

An assumption concerning the linear independence between gradients of constraint functions is presented as follows.
\begin{ass}
	\label{assu-4}
	$\left[\nabla \boldsymbol{h}, \nabla \boldsymbol{g}\right] \in \mathcal{R}^{n \times (m+p)}$ is a column full rank matrix, where
$\nabla \boldsymbol{h}$ and $ \nabla \boldsymbol{g}$ are defined in \eqref{notation}.
\end{ass}

Assumption \ref{assu-4} can be satisfied in certain situations. For example, resource constraint in a market is $\boldsymbol{W}:=\left\{\boldsymbol{y'} | \boldsymbol{A'y'} \leq \boldsymbol{b'} \right\}$, where $\boldsymbol{A'} \in \mathcal{R}^{m \times n}$ is a row full rank matrix and $\boldsymbol{b'} \in \mathcal{R}^{m}$. This assumption ensures that the loss function $l\left(\boldsymbol{\theta};\boldsymbol{y}, \boldsymbol{u}\right)$ defined by \eqref{newll} is both Lipschitz continuous and convex in the unknown parameter $\boldsymbol{\theta}$ for every $\boldsymbol{u} \in \boldsymbol{U}$ and $\boldsymbol{y}$. Then we provide two propositions.
\begin{pro} \label{lema2}
Under Assumptions \ref{convex}-\ref{assu-4}, the loss function (\ref{newll}) is uniformly $C$-Lipschitz continuous in $\Theta$ for every $\boldsymbol{u} \in \boldsymbol{U}$ and $\boldsymbol{y}$, namely,
\begin{align*}
\left|l({\boldsymbol{\theta}_1};\boldsymbol{y},\boldsymbol{u}) - l({\boldsymbol{\theta}_2};\boldsymbol{y},\boldsymbol{u}) \right| \leq C \left\| \boldsymbol{\theta}_1 - \boldsymbol{\theta}_2 \right\| _2, \forall \boldsymbol{\theta}_1, \boldsymbol{\theta}_2 \in  \boldsymbol{\Theta}.
\end{align*}
\end{pro}
\begin{proof}
  See Appendix A.
\end{proof}
\begin{pro}\label{conv}
Let Assumptions \ref{convex}-\ref{assu-4} hold. Then the loss function \eqref{newll} is convex in $\boldsymbol{\Theta}$ for every $\boldsymbol{u} \in \boldsymbol{U}$ and $\boldsymbol{y}$.
\end{pro}
\begin{proof}
    See Appendix B.
\end{proof}

\subsection{Regret Bound}
Under Assumption \ref{penalty}, the regularized loss function (\ref{update1}) can be rewritten as follows.
\begin{align} \label{gg}
	G_k(\boldsymbol{\theta}) = \tfrac{1}{2} \left\| \boldsymbol{\theta} - \boldsymbol{\theta}^{k} \right\|_2^2 + \mu _kl(\boldsymbol{\theta};\boldsymbol{y}^k,\boldsymbol{u}^k),
\end{align}
\noindent where $l(\boldsymbol{\theta};\boldsymbol{y}^k,\boldsymbol{u}^k)$ is defined by (\ref{newll}). The following theorem reveals that the regret bound of the online parameter identification algorithm is $O(\sqrt{K})$.

\begin{theorem}\label{thm-1}
Let Assumptions \ref{convex}-\ref{assu-4} hold. Consider the online parameter identification algorithm (Algorithm 1), where   $\mu_k = \mu_1/\sqrt{k},~\mu_1>0$.
 Then the regret bound defined in \eqref{regret} is $O(\sqrt{K})$.
\end{theorem}
\begin{proof}
Due to  \eqref{update_phi} and the convexity of the loss function \eqref{newll} from Proposition \ref{conv}, we can conclude that $\boldsymbol{0}\in \partial G_k(\tilde{\boldsymbol{\theta}}^{k+1})$, where $\partial G_k(\tilde{\boldsymbol{\theta}}^{k+1})$ is the subgradient set of $G_k(\cdot)$ at $\tilde{\boldsymbol{\theta}}^{k+1}$. Therefore, by recalling that  $D(\boldsymbol{x},\boldsymbol{y})=\frac{1}{2}\|\boldsymbol{x}-\boldsymbol{y}\|_2^2$ from Assumption \ref{penalty}, there exists a subgradient $\boldsymbol{s} \in \partial l(\tilde{\boldsymbol{\theta}}^{k+1};\boldsymbol{y}^k,\boldsymbol{u}^k)$, such that
	\begin{equation}
		\begin{aligned} \label{subgradient}
			\tilde{\boldsymbol{\theta}}^{k+1} - \boldsymbol{\theta}^k + \mu_k\boldsymbol{s} =\boldsymbol{0}.
		\end{aligned}
	\end{equation}
\noindent From the convexity of the loss function, we also have
\begin{equation} \label{conloss}
    \begin{aligned}
        l(\boldsymbol{\theta}_*^K;\boldsymbol{y}^k,\boldsymbol{u}^k) \geq l(\tilde{\boldsymbol{\theta}}^{k+1};\boldsymbol{y}^k,\boldsymbol{u}^k) + \boldsymbol{s}^T(\boldsymbol{\theta}_*^K-\tilde{\boldsymbol{\theta}}^{k+1}).
    \end{aligned}
\end{equation}
By using (\ref{subgradient}), we get
\begin{equation} \label{qw14}
	\begin{aligned}	
    &-\mu_k\boldsymbol{s}^T\left(\boldsymbol{\theta}_*^K-\tilde{\boldsymbol{\theta}}^{k+1}\right)
		= \left(\tilde{\boldsymbol{\theta}}^{k+1} - \boldsymbol{\theta}^k\right)^T\left(\boldsymbol{\theta}_*^K-\tilde{\boldsymbol{\theta}}^{k+1}\right) \\
		=& \frac{1}{2} \left\| \boldsymbol{\theta}^{k} - \boldsymbol{\theta}_*^K \right\|_2^2  - \frac{1}{2} \left\| \tilde{\boldsymbol{\theta}}^{k+1} - \boldsymbol{\theta}_*^K \right\|_2^2 - \frac{1}{2} \left\| \tilde{\boldsymbol{\theta}}^{k+1} - \boldsymbol{\theta}^{k} \right\|_2^2.
	\end{aligned}
\end{equation}

According to the Lipschitz continuity of the loss function from Proposition \ref{lema2}, we have
\begin{equation} \label{418}
	\begin{aligned}
		& \mu_k\left(l\left(\boldsymbol{\theta}^k;\boldsymbol{y}^k, \boldsymbol{u}^k\right) - l\left(\tilde{\boldsymbol{\theta}}^{k+1};\boldsymbol{y}^k, \boldsymbol{u}^k\right)\right) - \tfrac{1}{2} \left\| \tilde{\boldsymbol{\theta}}^{k+1} - \boldsymbol{\theta}^{k} \right\|_2^2 \\
		\leq & C \mu_k \left\| \boldsymbol{\theta}^k - \tilde{\boldsymbol{\theta}}^{t+1} \right\|_2 - \tfrac{1}{2} \left\| \tilde{\boldsymbol{\theta}}^{k+1} - \boldsymbol{\theta}^{k} \right\|_2^2 \\
		= & 2\left[\tfrac{1}{2}\left\| \boldsymbol{\theta}^k - \tilde{\boldsymbol{\theta}}^{k+1} \right\|_2 (C \mu_k - \tfrac{1}{2}\left\| \boldsymbol{\theta}^k - \tilde{\boldsymbol{\theta}}^{k+1} \right\|_2)\right] \\
		\leq & 2\left[\frac{\tfrac{1}{2}\left\| \boldsymbol{\theta}^k - \tilde{\boldsymbol{\theta}}^{k+1} \right\|_2 + C \mu_k - \frac{1}{2}\left\| \boldsymbol{\theta}^k - \tilde{\boldsymbol{\theta}}^{k+1} \right\|_2}{2}\right]^2	=   \frac{C^2\mu _k^2}{2},
	\end{aligned}
\end{equation}
where the last inequality holds from $ab \leq (\frac{a+b}{2})^2$ for every real numbers $a,b$.

 Because $\boldsymbol{\theta}^{k+1}$ is the projection of $\tilde{\boldsymbol{\theta}}^{k+1}$ onto the closed and convex set $\boldsymbol{\Theta}$, and $\boldsymbol{\theta}_*^K \in \boldsymbol{\Theta}$, according to the generalized Pythagorous inequality (see \cite[Theorem 2.5.1]{censor}), we have $\frac{1}{2} \left\| \tilde{\boldsymbol{\theta}}^{k+1} - \boldsymbol{\theta}^{k+1} \right\|_2^2 \leq \frac{1}{2} \left\| \tilde{\boldsymbol{\theta}}^{k+1} - \boldsymbol{\theta}_*^K \right\|_2^2 - \frac{1}{2} \left\| \boldsymbol{\theta}^{k+1} - \boldsymbol{\theta}_*^K \right\|_2^2$.
Thus,
\begin{equation} \label{pyi}
	\begin{aligned}
	&\frac{1}{2} \left\| \boldsymbol{\theta}^{k} - \boldsymbol{\theta}_*^K \right\|_2^2 - \frac{1}{2} \left\| \tilde{\boldsymbol{\theta}}^{k+1} - \boldsymbol{\theta}_*^K \right\|_2^2 \\
	\leq& \frac{1}{2} \left\| \boldsymbol{\theta}^{k} - \boldsymbol{\theta}_*^K \right\|_2^2 - \frac{1}{2} \left\| \tilde{\boldsymbol{\theta}}^{k+1} - \boldsymbol{\theta}_*^K \right\|_2^2 + \frac{1}{2} \left\| \tilde{\boldsymbol{\theta}}^{k+1} - \boldsymbol{\theta}^{k+1} \right\|_2^2 \\
	\leq& \frac{1}{2} \left\| \boldsymbol{\theta}^{k} - \boldsymbol{\theta}_*^K \right\|_2^2 - \frac{1}{2} \left\| \boldsymbol{\theta}^{k+1} - \boldsymbol{\theta}_*^K \right\|_2^2.
	\end{aligned}
\end{equation}
Therefore,
\begin{equation} \label{17qqq}
	\begin{aligned}
		&\mu_kl(\boldsymbol{\theta}^k;\boldsymbol{y}^k,\boldsymbol{u}^k) - \mu_kl(\boldsymbol{\theta}_*^K;\boldsymbol{y}^k, \boldsymbol{u}^k) \\
		\overset{(\ref{conloss})}{\leq}& \mu_kl(\boldsymbol{\theta}^k;\boldsymbol{y}^k,\boldsymbol{u}^k) - \mu_kl(\tilde{\boldsymbol{\theta}}^{k+1};\boldsymbol{y}^k,\boldsymbol{u}^k) - \mu_k\boldsymbol{s}^T(\boldsymbol{\theta}_*^K-\tilde{\boldsymbol{\theta}}^{k+1}) \\
		\overset{(\ref{qw14})}{=}& \mu_kl(\boldsymbol{\theta}^k;\boldsymbol{y}^k,\boldsymbol{u}^k) - \mu_kl(\tilde{\boldsymbol{\theta}}^{k+1};\boldsymbol{y}^k,\boldsymbol{u}^k) + \tfrac{1}{2} \left\| \boldsymbol{\theta}^{k} - \boldsymbol{\theta}_*^K \right\|_2^2
\\ 	&- \tfrac{1}{2} \left\| \tilde{\boldsymbol{\theta}}^{k+1} - \boldsymbol{\theta}_*^K \right\|_2^2 - \tfrac{1}{2} \left\| \tilde{\boldsymbol{\theta}}^{k+1} - \boldsymbol{\theta}^{k} \right\|_2^2 \\
		\overset{(\ref{418})}{\leq}& \tfrac{C^2\mu _k^2}{2} + \tfrac{1}{2} \left\| \boldsymbol{\theta}^{k} - \boldsymbol{\theta}_*^K \right\|_2^2 - \tfrac{1}{2} \left\| \tilde{\boldsymbol{\theta}}^{k+1} - \boldsymbol{\theta}_*^K \right\|_2^2  \\
		\overset{(\ref{pyi})}{\leq}& \tfrac{C^2\mu _k^2}{2} + \tfrac{1}{2} \left\| \boldsymbol{\theta}^{k} - \boldsymbol{\theta}_*^K \right\|_2^2 - \tfrac{1}{2} \left\| \boldsymbol{\theta}^{k+1} - \boldsymbol{\theta}_*^K \right\|_2^2.
	\end{aligned}
\end{equation}
After dividing both sides of inequality (19) by $\mu_k$ and summing over $k$, we obtain the following result:
\begin{equation}
	\begin{aligned} \label{123}
		R_K =&  \sum_{k=1}^{K} l(\boldsymbol{\theta}^k;\boldsymbol{y}^k, \boldsymbol{u}^k) - \sum_{k=1}^{K} l(\boldsymbol{\theta}_*^K;\boldsymbol{y}^k, \boldsymbol{u}^k) \\
		= & \underbrace{\sum_{k=1}^K \tfrac{C^2}{2} \mu_k}_{\mathrm{\uppercase\expandafter{\romannumeral1}}} + \underbrace{\sum_{k=1}^K \tfrac{1}{2\mu_k}\left(\left\|\boldsymbol{\theta}^{k}-\boldsymbol{\theta}_*^K\right\|_2^2 - \left\|\boldsymbol{\theta}^{k+1}-\boldsymbol{\theta}_*^K\right\|_2^2\right)}_{\mathrm{\uppercase\expandafter{\romannumeral2}}}.
	\end{aligned}
\end{equation}
Recall $\mu_k = \mu_1/\sqrt{k}$. For the part $\mathrm{\uppercase\expandafter{\romannumeral1}}$, using $\sum_{k=1}^{K} \frac{1}{\sqrt{k}} \leq 1+ \int_{1}^{K}\frac{1}{\sqrt{k}}dt \leq 2 \sqrt{K}$, we have
$\sum_{k=1}^K \frac{C^2}{2} \mu_k  \leq \mu_1C^2 \sqrt{K}.$
For the part $\mathrm{\uppercase\expandafter{\romannumeral2}}$, since $\|\theta\|_2 <B_1$ from Assumption \ref{assu-2}, we have
\begin{align*}
	& \sum_{k=1}^K \frac{1}{2\mu_k}\left(\left\|\boldsymbol{\theta}^{k}-\boldsymbol{\theta}_*^K\right\|_2^2 - \left\|\boldsymbol{\theta}^{k+1}-\boldsymbol{\theta}_*^K\right\|_2^2\right) \\
	=& \frac{1}{2\mu_1} \sum_{k=1}^K \sqrt{k}\left(\left\|\boldsymbol{\theta}^{k}-\boldsymbol{\theta}_*^K\right\|_2^2 - \left\|\boldsymbol{\theta}^{k+1}-\boldsymbol{\theta}_*^K\right\|_2^2\right) \\
	=& \frac{1}{2\mu_1}  \sum_{k=2}^K \left\|\boldsymbol{\theta}^{k}-\boldsymbol{\theta}_*^K\right\|_2^2 (\sqrt{k}-\sqrt{k-1})  \\
	&+ \frac{1}{2\mu_1} \left[ \left\|\boldsymbol{\theta}^{1}-\boldsymbol{\theta}_*^K\right\|_2^2 - \sqrt{K} \left\|\boldsymbol{\theta}^{K+1}-\boldsymbol{\theta}_*^K\right\|_2^2  \right] \\
	\leq& \frac{1}{\mu_1} \left[2B_1^2(\sqrt{K}-1) + 2B_1^2\right]
	=  \frac{1}{\mu_1} 2B_1^2 \sqrt{K}.
\end{align*}
Therefore, $ R_K \leq \left(\frac{2B_1^2}{\mu_1} + \mu_1C^2 \right) \sqrt{K}$ from \eqref{123}. Thus, the conclusion follows.
\end{proof}

\section{Numerical Simulations} \label{s5}

This section aims to apply the proposed Algorithm  1 to identify the parameters of cost functions in a Nash-Cournot model, which is a widely adopted framework for modeling market competition and can be seen as a GNCG \cite[]{gnep4,gnep7,Salant,Abada}.

\subsection{Simulation Model}
Consider an energy market with $N$ companies engaged in selling natural gas. Each company $v$ competes in the market by determining its output $x_v$. When $x_v\geq0$, company $v$ sells $x_v$ units of natural gas to the market; when $x_v<0$, it buys $|x_v|$ units of natural gas from the market.  Additionally, a minimum gas demand of $q_k>0$ must be met, i.e., $\sum_{v=1}^{N}x_v\geq q_k$. Suppose that
the market price of natural gas is influenced by the total amount of natural gas $\sum_{v=1}^{N}x_v$ in the market, and is set as  $p_k := a_k + b_k \sum_{v=1}^{N}x_v,$ where $a_k>0, b_k<0$ are observable variables affected by the market. Take $a_k+b_kq_k>0$ to make the price meaningful. The profit of company $v$ is given by $f_v\left(x_v, \mathbf{x}_{-v},(a_k,b_k),\theta_v\right) := p_kx_v - \theta_vx_v$, where $\theta_v\geq0$ denotes company $v$'s production cost per unit of natural gas. The goal of every company $v$ is to maximize its profit, namely, it select a strategy $x_v$ to minimize the following optimization problem:
\begin{equation*}
	\begin{aligned}
		&\min_{x_v} - f_v\left(x_v, \mathbf{x}_{-v},(a_k,b_k),\theta_v\right), \quad s.t. \ \sum_{v=1}^{N}x_v \geq q_k.
	\end{aligned}
\end{equation*}
This game problem is a jointly convex GNCG. Let $(a_k, b_k, q_k)$ be an observable signal that adjust with the  changes of the   market, and $\theta_v \in \{\theta_v:\theta_v\geq0\}, v \in \mathcal{N}$ be unknown parameters that we aim to estimate online.

\subsection{Simulation Setting and Results}
Consider a market with three companies, where the unknown parameter vector $\boldsymbol{\theta}=\left[\theta_1,\theta_2,\theta_3\right]^T$ is taken as $\left[10,7.5,6\right]^T$. In each round $k$, we independently sample $\boldsymbol{u}^k=\left(a_k, b_k, q_k\right)$ from uniform distributions with ranges $\left[15,15\times120\right]$, $\left[1,1\times120\right]$, and $\left[5,5\times120\right]$, respectively. The observed equilibrium result  $\boldsymbol{y}^k=\boldsymbol{x}^k+\boldsymbol{\epsilon}^k$  incorporates noise, where $\boldsymbol{x}^k$ represents the Nash-Cournot equilibrium under the signal $\boldsymbol{u}^k$, and $\boldsymbol{\epsilon}^k$ is a random vector drawn from the multivariate standard normal distribution. Subsequently, we set $K=100$ and run the online parameter estimation algorithm and the algorithm in the whole batch setting. The effectiveness of this online algorithm is evaluated using the following three performance metrics.
\begin{enumerate}
	\item Time(s):  the time required for the algorithm per round.
	\item $R_k/k$: the average regret of the algorithm. 
	\item $\|\boldsymbol{\theta}^k-\boldsymbol{\theta}_*^k\|_2$:  the deviation between the estimates $\boldsymbol{\theta}^k$ and $\boldsymbol{\theta}_*^k$. Here, $\boldsymbol{\theta}^k$ denotes the estimated unknown parameter in the online parameter algorithm at the $k$-th round, while $\boldsymbol{\theta}_*^k$ represents the corresponding estimate by using data from the previous $k$ rounds in the whole batch setting. 
\end{enumerate}

By setting $\mu_k = \mu_1/\sqrt{k}$ with $\mu_1=0.1$ in Algorithm \ref{alg1}. A comparative analysis of the results obtained with the whole batch setting is illustrated in Figure \ref{main_plot}. Notably, the online algorithm demonstrates a significant advantage in terms of execution time, which remains short and efficient. Conversely, under the whole batch setting, the execution time grows almost linearly with the volume of data. Moreover, both the average regret and $\|\boldsymbol{\theta}^k-\boldsymbol{\theta}_*^k\|_2$ exhibit almost the same  convergence speed towards 0, indicating that the performance  of  Algorithm 1 closely approximates that of the whole batch setting. These results highlight the effectiveness and efficiency of the online algorithm.

\vskip -4mm
\begin{figure}[h]
	\centering 
	\includegraphics[height=9cm,width=6cm]{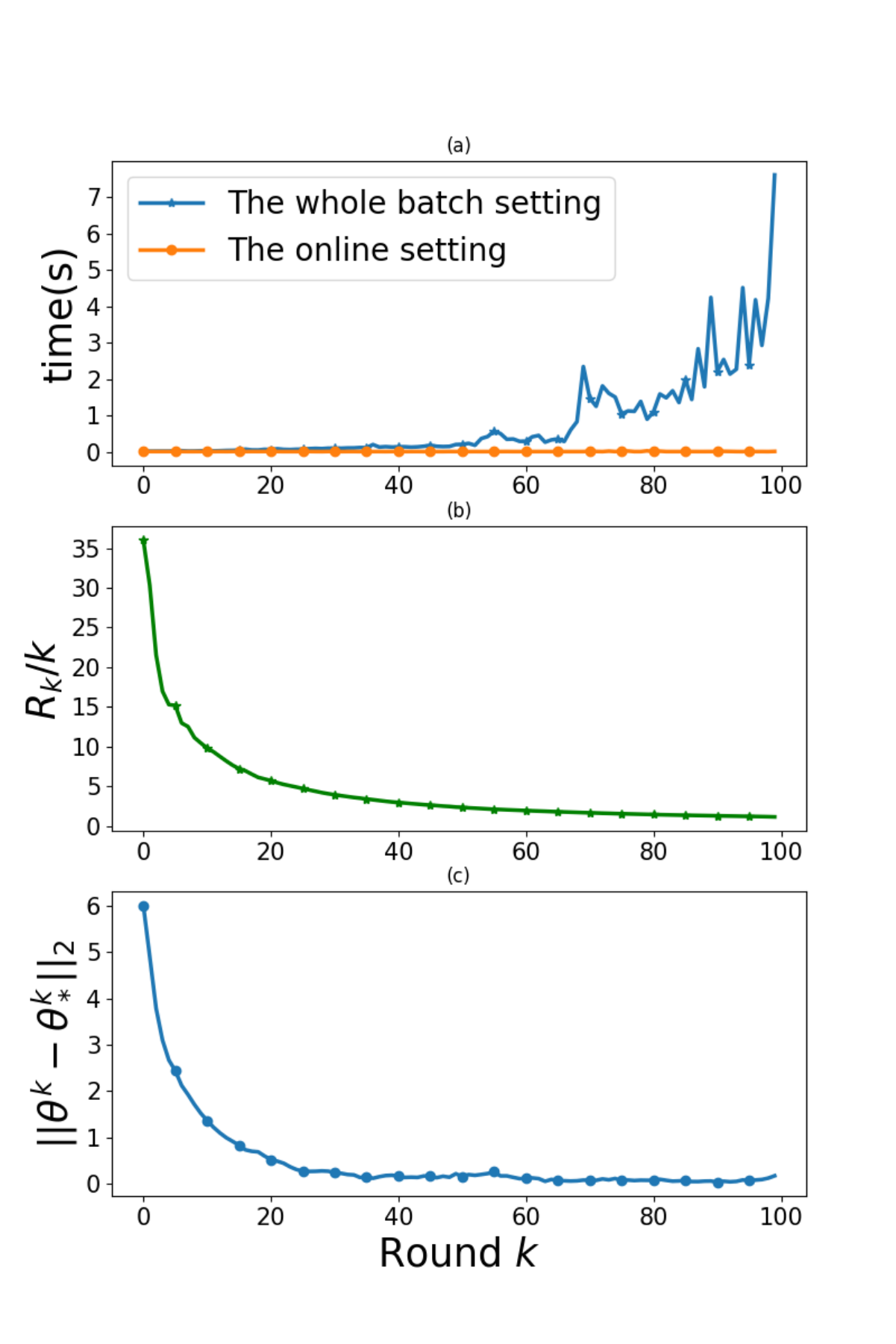}
	\caption{Online parameter identification results in a natural gas market. (a) The vertical coordinate is the time of algorithm execution. (b) The vertical coordinate indicates the average regret per round. (c) The vertical coordinate indicates the deviation between $\boldsymbol{\theta}^k$ and $\boldsymbol{\theta}_*^k$ for each round.}
	\label{main_plot}
\end{figure}

The learning rate $\mu_k$ plays an important role in Algorithm 1 by balancing the existing information with  new observations. We set different learning rates  $\mu_k = \mu_1/\sqrt{k}$ with $\mu_1=0.1, 0.3$ and $0.5$ to explore their impact on the algorithm. The empirical results are presented in Figure \ref{regret_plot} and Figure \ref{est_plot}.  Figure \ref{regret_plot} illustrates that a higher learning rate leads to a lower regret bound. This is might because a larger learning rate indicates a greater emphasis on new observations. Besides,  Figure \ref{est_plot} demonstrates that a larger learning rate facilitates faster convergence of $\boldsymbol{\theta}^k$ to $\boldsymbol{\theta}_*^k$, albeit with greater fluctuations.

\vspace{-5mm}

\begin{figure}
	\centering
	\includegraphics[height=5cm,width=6.6cm]{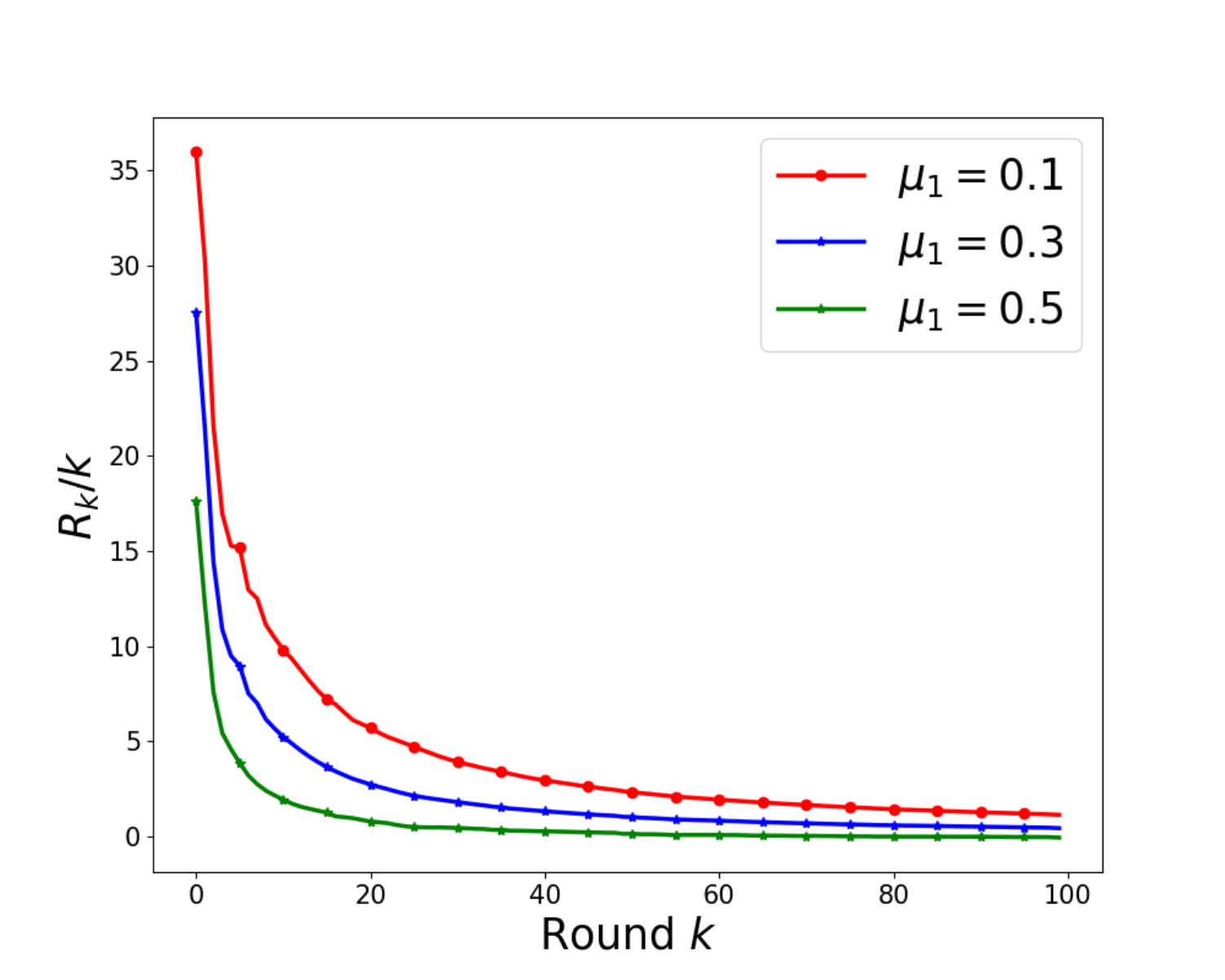}
	\caption{The effect of  learning rates on the average regret.}
	\label{regret_plot}
\end{figure}
\begin{figure}
	\centering
	\includegraphics[height=5cm,width=6.6cm]{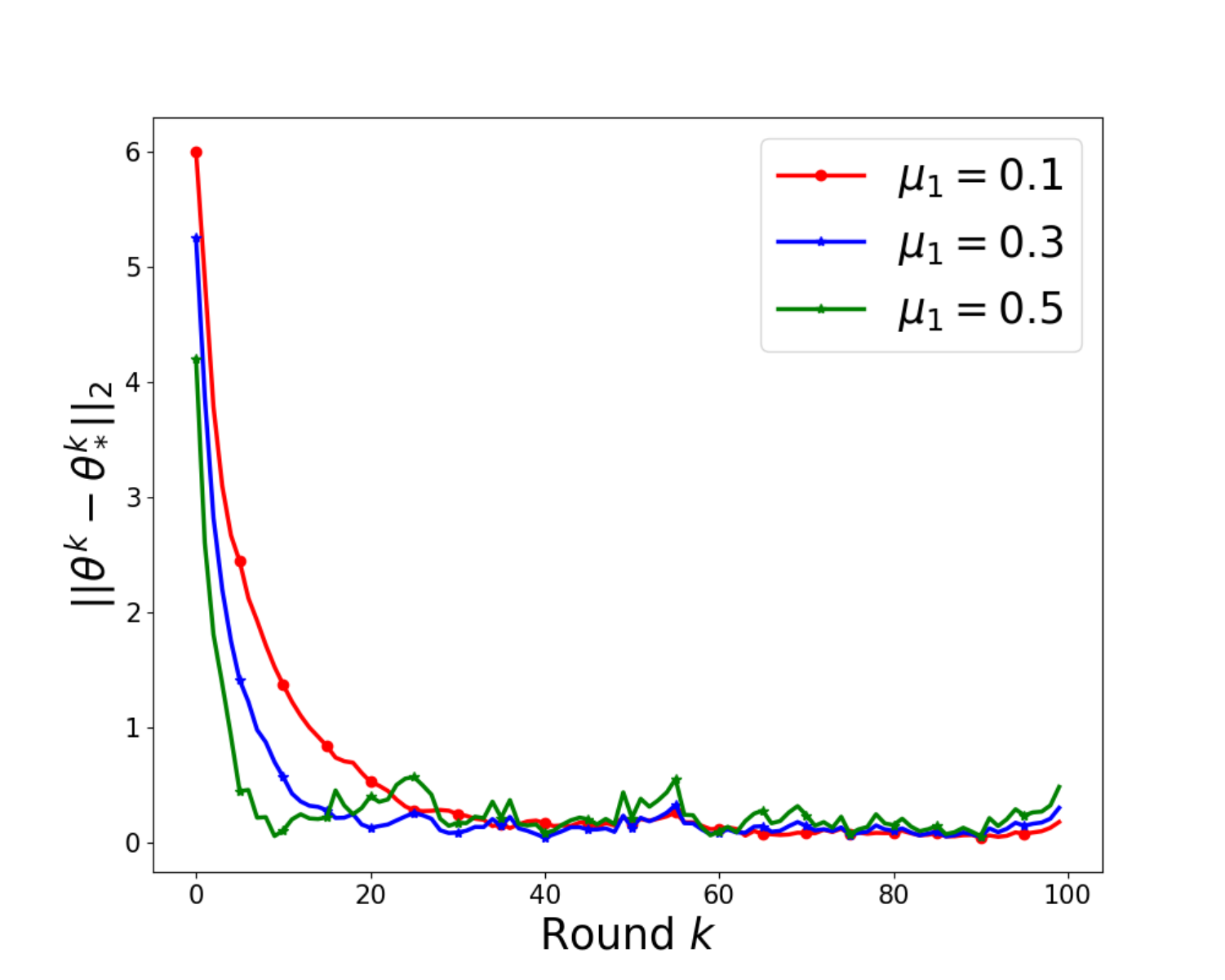}
	\caption{The impact of   learning rates on the estimation deviation  $\| \boldsymbol{\theta}^k-\boldsymbol{\theta}_*^k\|_2$.}
	\label{est_plot}
\end{figure}

\section{Conclusion} \label{s6}
In this work, an online  algorithm was designed  to identify the parameters in the cost functions of GNCG, and it was proven that when the cost function was linear in the unknown parameters, and the learning rate of the algorithm $\mu_k \propto 1/\sqrt{k}$ along with other assumptions were satisfied, the regret bound was $O(\sqrt{K})$. Numerical simulations of a Nash-Cournot problem were also implemented to demonstrate that the performance of the online algorithm was comparable to the whole batch setting after some rounds.



\bibliographystyle{unsrt}
\bibliography{ref}

\begin{thebibliography}{10}

\bibitem{Roy}
Sankardas Roy, Charles Ellis, Sajjan Shiva, Dipankar Dasgupta, Vivek Shandilya,
  and Qishi Wu.
\newblock A survey of game theory as applied to network security.
\newblock In {\em 2010 43rd Hawaii International Conference on System
  Sciences}, pages 1--10. IEEE, 2010.

\bibitem{Alvarez}
Israel Alvarez and Alexander Poznyak.
\newblock Game theory applied to urban traffic control problem.
\newblock In {\em ICCAS 2010}, pages 2164--2169. IEEE, 2010.

\bibitem{mei}
Shengwei Mei, Wei Wei, and Feng Liu.
\newblock On engineering game theory with its application in power systems.
\newblock {\em Control Theory and Technology}, 15:1--12, 2017.

\bibitem{genp1}
Gerard Debreu.
\newblock A social equilibrium existence theorem.
\newblock {\em Proceedings of the National Academy of Sciences},
  38(10):886--893, 1952.

\bibitem{gnep2}
Patrick~T Harker.
\newblock Generalized nash games and quasi-variational inequalities.
\newblock {\em European journal of Operational research}, 54(1):81--94, 1991.

\bibitem{gnep3}
Chung-Kai Yu, Mihaela Van Der~Schaar, and Ali~H Sayed.
\newblock Distributed learning for stochastic generalized nash equilibrium
  problems.
\newblock {\em IEEE Transactions on Signal Processing}, 65(15):3893--3908,
  2017.

\bibitem{gnep4}
Javier Contreras, Matthias Klusch, and Jacek~B Krawczyk.
\newblock Numerical solutions to nash-cournot equilibria in coupled constraint
  electricity markets.
\newblock {\em IEEE Transactions on Power Systems}, 19(1):195--206, 2004.

\bibitem{gnep5}
Andreas Fischer, Markus Herrich, and Klaus Sch{\"o}nefeld.
\newblock Generalized nash equilibrium problems-recent advances and challenges.
\newblock {\em Pesquisa Operacional}, 34:521--558, 2014.

\bibitem{gnep7}
Benjamin~F Hobbs and Jong-Shi Pang.
\newblock Nash-cournot equilibria in electric power markets with piecewise
  linear demand functions and joint constraints.
\newblock {\em Operations Research}, 55(1):113--127, 2007.

\bibitem{gnep8}
Benjamin~F Hobbs and Jong-Shi Pang.
\newblock Nash-cournot equilibria in electric power markets with piecewise
  linear demand functions and joint constraints.
\newblock {\em Operations Research}, 55(1):113--127, 2007.

\bibitem{fran}
Francisco Facchinei and Christian Kanzow.
\newblock Generalized nash equilibrium problems.
\newblock {\em Annals of Operations Research}, 175(1):177--211, 2010.

\bibitem{fangfei}
Jibang Wu, Weiran Shen, Fei Fang, and Haifeng Xu.
\newblock Inverse game theory for stackelberg games: the blessing of bounded
  rationality.
\newblock {\em Advances in Neural Information Processing Systems},
  35:32186--32198, 2022.

\bibitem{Le}
Simon Le~Cleac’h, Mac Schwager, and Zachary Manchester.
\newblock Lucidgames: Online unscented inverse dynamic games for adaptive
  trajectory prediction and planning.
\newblock {\em IEEE Robotics and Automation Letters}, 6(3):5485--5492, 2021.

\bibitem{online2}
Yuan Gao, Alex Peysakhovich, and Christian Kroer.
\newblock Online market equilibrium with application to fair division.
\newblock {\em Advances in Neural Information Processing Systems},
  34:27305--27318, 2021.

\bibitem{linear1}
Arunabha Bagchi and Vivek Borkar.
\newblock Parameter identification in infinte dimensional linear systems:
  Parameter identification.
\newblock {\em Stochastics: An International Journal of Probability and
  Stochastic Processes}, 12(3-4):201--213, 1984.

\bibitem{linear2}
Crist{\'o}v{\~a}o~D Sousa and Rui Cortesao.
\newblock Physical feasibility of robot base inertial parameter identification:
  A linear matrix inequality approach.
\newblock {\em The International Journal of Robotics Research}, 33(6):931--944,
  2014.

\bibitem{linear3}
Jing Na, Xuemei Ren, and Yuanqing Xia.
\newblock Adaptive parameter identification of linear siso systems with unknown
  time-delay.
\newblock {\em Systems \& Control Letters}, 66:43--50, 2014.

\bibitem{nonlinear1}
Eleni~N Chatzi, Andrew~W Smyth, and Sami~F Masri.
\newblock Experimental application of on-line parametric identification for
  nonlinear hysteretic systems with model uncertainty.
\newblock {\em Structural Safety}, 32(5):326--337, 2010.

\bibitem{nonlinear2}
Roger Ghanem and Francesco Romeo.
\newblock A wavelet-based approach for model and parameter identification of
  non-linear systems.
\newblock {\em International Journal of Non-Linear Mechanics}, 36(5):835--859,
  2001.

\bibitem{nonlinear3}
Andrew~W Smyth, Sami~F Masri, Anastasios~G Chassiakos, and Thomas~K Caughey.
\newblock On-line parametric identification of mdof nonlinear hysteretic
  systems.
\newblock {\em Journal of engineering mechanics}, 125(2):133--142, 1999.

\bibitem{Molloy}
Timothy~L Molloy, Jairo Inga~Charaja, S{\"o}ren Hohmann, and Tristan Perez.
\newblock Inverse noncooperative dynamic games.
\newblock In {\em Inverse Optimal Control and Inverse Noncooperative Dynamic
  Game Theory: A Minimum-Principle Approach}, pages 143--187. Springer, 2022.

\bibitem{Cao}
Kun Cao and Lihua Xie.
\newblock Game-theoretic inverse reinforcement learning: A differential
  pontryagin's maximum principle approach.
\newblock {\em IEEE Transactions on Neural Networks and Learning Systems},
  2022.

\bibitem{molloy2019}
Timothy~L Molloy, Jairo Inga, Michael Flad, Jason~J Ford, Tristan Perez, and
  S{\"o}ren Hohmann.
\newblock Inverse open-loop noncooperative differential games and inverse
  optimal control.
\newblock {\em IEEE Transactions on Automatic Control}, 65(2):897--904, 2019.

\bibitem{rothfu2017}
Simon Rothfu{\ss}, Jairo Inga, Florian K{\"o}pf, Michael Flad, and S{\"o}ren
  Hohmann.
\newblock Inverse optimal control for identification in non-cooperative
  differential games.
\newblock {\em IFAC-PapersOnLine}, 50(1):14909--14915, 2017.

\bibitem{Allen}
Stephanie Allen, Steven~A Gabriel, and John~P Dickerson.
\newblock Using inverse optimization to learn cost functions in generalized
  nash games.
\newblock {\em Computers \& Operations Research}, 142:105721, 2022.

\bibitem{yu2022}
Yue Yu, Jonathan Salfity, David Fridovich-Keil, and Ufuk Topcu.
\newblock Inverse matrix games with unique quantal response equilibrium.
\newblock {\em IEEE Control Systems Letters}, 7:643--648, 2022.

\bibitem{lin2017}
Xiaomin Lin, Peter~A Beling, and Randy Cogill.
\newblock Multiagent inverse reinforcement learning for two-person zero-sum
  games.
\newblock {\em IEEE Transactions on Games}, 10(1):56--68, 2017.

\bibitem{Zhang}
Zhenhua Zhang, Yao Li, and Chengpu Yu.
\newblock Online inverse identification of noncooperative dynamic games.
\newblock In {\em 2021 IEEE International Conference on Unmanned Systems
  (ICUS)}, pages 408--413. IEEE, 2021.

\bibitem{gnep6}
J~Ben Rosen.
\newblock Existence and uniqueness of equilibrium points for concave n-person
  games.
\newblock {\em Econometrica: Journal of the Econometric Society}, pages
  520--534, 1965.

\bibitem{Kulkarni}
Ankur~A Kulkarni and Uday~V Shanbhag.
\newblock On the variational equilibrium as a refinement of the generalized
  nash equilibrium.
\newblock {\em Automatica}, 48(1):45--55, 2012.

\bibitem{Kulis2010}
Brian Kulis and Peter~L Bartlett.
\newblock Implicit online learning.
\newblock In {\em Proceedings of the 27th International Conference on Machine
  Learning (ICML-10)}, pages 575--582, 2010.

\bibitem{calafiore2014optimization}
Giuseppe~C Calafiore and Laurent El~Ghaoui.
\newblock {\em Optimization models}.
\newblock Cambridge university press, 2014.

\bibitem{nocedal1999numerical}
Jorge Nocedal and Stephen~J Wright.
\newblock {\em Numerical optimization}.
\newblock Springer, 1999.

\bibitem{li2017}
Nan Li, Dave~W Oyler, Mengxuan Zhang, Yildiray Yildiz, Ilya Kolmanovsky, and
  Anouck~R Girard.
\newblock Game theoretic modeling of driver and vehicle interactions for
  verification and validation of autonomous vehicle control systems.
\newblock {\em IEEE Transactions on control systems technology},
  26(5):1782--1797, 2017.

\bibitem{censor}
Yair Censor and Stavros~Andrea Zenios.
\newblock {\em Parallel optimization: Theory, algorithms, and applications}.
\newblock Oxford University Press, USA, 1997.

\bibitem{Salant}
Stephen~W Salant.
\newblock Imperfect competition in the international energy market: a
  computerized nash-cournot model.
\newblock {\em Operations research}, 30(2):252--280, 1982.

\bibitem{Abada}
Ibrahim Abada, Steven Gabriel, Vincent Briat, and Olivier Massol.
\newblock A generalized nash--cournot model for the northwestern european
  natural gas markets with a fuel substitution demand function: The gammes
  model.
\newblock {\em Networks and Spatial Economics}, 13:1--42, 2013.

\end{thebibliography}





\appendix

\section{Proof of Proposition \ref{lema2}} 
To prove Proposition \ref{lema2}, we first give the following lemma.
\begin{lemma}
For every $\boldsymbol{\theta}_1 \in \boldsymbol{\Theta}$ and $\boldsymbol{\theta}_2  \in \boldsymbol{\Theta}$, let $\left(\boldsymbol{\lambda}^{\ast },\boldsymbol{\nu}^{\ast}\right)$ and $\left(\boldsymbol{\lambda}^{\# },\boldsymbol{\nu}^{\#}\right)$ be the optimal solutions of the optimization problem
\begin{equation} \label{op_pro}
    \begin{aligned}
        \min_{\boldsymbol{\lambda} \geq  \boldsymbol{0},\boldsymbol{\nu}} L\left(\boldsymbol{\theta}, \boldsymbol{\lambda}, \boldsymbol{\nu};\boldsymbol{y}, \boldsymbol{u}\right),
    \end{aligned}
\end{equation}
corresponding to $\boldsymbol{\theta}_1$ and $\boldsymbol{\theta}_2$, respectively. Under Assumptions \ref{convex} - \ref{assu-4}, for every $\boldsymbol{u} \in \boldsymbol{U}$ and $\boldsymbol{y}$, there exists two real numbers $a, d>0$  such that $\left\| \boldsymbol{\lambda}^{\ast} - \boldsymbol{\lambda}^{\#} \right\|_2 \leq  a \left\|\boldsymbol{\theta}_1 - \boldsymbol{\theta}_2 \right\|_2, \left\| \boldsymbol{\nu}^{\ast} - \boldsymbol{\nu}^{\#} \right\|_2 \leq  a \left\|\boldsymbol{\theta}_1 - \boldsymbol{\theta}_2 \right\|_2$ and $\left\| \boldsymbol{\lambda}^{\ast} \right\| _2,  \left\| \boldsymbol{\lambda}^{\#}  \right\| _2 , \left\| \boldsymbol{\nu}^{\ast} \right\| _2,  \left\| \boldsymbol{\nu}^{\#}  \right\| _2 \leq d$.
\label{lem-1}
\end{lemma}
\begin{proof}
By recalling the definition of $L\left(\boldsymbol{\theta}, \boldsymbol{\lambda}, \boldsymbol{\nu};\boldsymbol{y}, \boldsymbol{u}\right)$ under \eqref{newll} and the notation of (\ref{notation}), we have
\begin{equation} \label{newL}
    \begin{aligned}
        L\left(\boldsymbol{\theta}, \boldsymbol{\lambda}, \boldsymbol{\nu};\boldsymbol{y}, \boldsymbol{u}\right)
        =& \left\|\boldsymbol{F}_{\boldsymbol{\theta}} + \nabla \boldsymbol{h}\boldsymbol{\lambda} + \nabla \boldsymbol{g} \boldsymbol{\nu}\right\|_2^2 + \left\|\boldsymbol{H}\boldsymbol{\lambda}\right\|_2^2 + \|\boldsymbol{g}\|_2^2\\
        =& \begin{bmatrix}
              \boldsymbol{\lambda}   \\
              \boldsymbol{\nu}  \\
            \end{bmatrix}^T
            \begin{bmatrix}
              \nabla \boldsymbol{h}^T \nabla \boldsymbol{h} + \boldsymbol{H}^2 &  \nabla \boldsymbol{h}^T \nabla \boldsymbol{g} \\
              \nabla \boldsymbol{g}^T \nabla \boldsymbol{h} &  \nabla \boldsymbol{g}^T \nabla \boldsymbol{g} \\
            \end{bmatrix}
            \begin{bmatrix}
              \boldsymbol{\lambda}   \\
              \boldsymbol{\nu}  \\
            \end{bmatrix} \\
        &+ 2 \boldsymbol{F}_{\boldsymbol{\theta}}^T\begin{bmatrix}
              \nabla \boldsymbol{h} & \nabla \boldsymbol{g}
            \end{bmatrix}
            \begin{bmatrix}
              \boldsymbol{\lambda}   \\
              \boldsymbol{\nu}  \\
            \end{bmatrix} + \boldsymbol{F}_{\boldsymbol{\theta}}^T\boldsymbol{F}_{\boldsymbol{\theta}} + \boldsymbol{g}^T\boldsymbol{g},
    \end{aligned}
\end{equation}
 where $\boldsymbol{H}^2 =diag\left(\boldsymbol{h}_1^2(\boldsymbol{y}, \boldsymbol{u}),\cdots,\boldsymbol{h}_m^2(\boldsymbol{y}, \boldsymbol{u})\right)$. Therefore, the optimization problem (\ref{op_pro}) can be abbreviated as
      \begin{equation}\label{zAz}
        \begin{aligned}
          \min_{\boldsymbol{z}}& ~\boldsymbol{z}^T\boldsymbol{Az} + 2\boldsymbol{F}_{\boldsymbol{\theta}}^T\boldsymbol{bz} + \tilde{\boldsymbol{c}} ,\ s.t. \ \boldsymbol{Bz} \geq \boldsymbol{0},
        \end{aligned}
      \end{equation}
where $\boldsymbol{z}=\begin{bmatrix}
            \boldsymbol{\lambda}   \\
            \boldsymbol{\nu}  \\
           \end{bmatrix}$, $\boldsymbol{A}=\begin{bmatrix}
          \nabla \boldsymbol{h}^T \nabla \boldsymbol{h} + \boldsymbol{H}^2 &  \nabla \boldsymbol{h}^T \nabla \boldsymbol{g} \\
          \nabla \boldsymbol{g}^T \nabla \boldsymbol{h} &  \nabla \boldsymbol{g}^T \nabla \boldsymbol{g} \\
           \end{bmatrix}$, $\boldsymbol{b}=\begin{bmatrix}
          \nabla \boldsymbol{h} & \nabla \boldsymbol{g}
          \end{bmatrix}$, $\tilde{\boldsymbol{c}}=\boldsymbol{F}_{\boldsymbol{\theta}}^T\boldsymbol{F}_{\boldsymbol{\theta}} + \boldsymbol{g}^T\boldsymbol{g}$ and $\boldsymbol{B}=\begin{bmatrix}
          \boldsymbol{I}_m & \boldsymbol{O}_p
          \end{bmatrix}$.
We now show that matrix $\boldsymbol{A}$ is a positive definite matrix. Because $\left[\nabla \boldsymbol{h}, \nabla \boldsymbol{g}\right]$ is a column full rank matrix from Assumption \ref{assu-4}, we have that $\nabla \boldsymbol{h w}_1 + \nabla \boldsymbol{g w}_2 \neq 0$ for every non-zero vector $\boldsymbol{w} = \left[\boldsymbol{w}_1^T,\boldsymbol{w}_2^T\right]^T \in \mathcal{R}^{m+p}$. Thus, $\boldsymbol{w}^T\boldsymbol{Aw} = \left\| \nabla \boldsymbol{h w}_1 + \nabla \boldsymbol{g w}_2 \right\|_2^2 + \| \boldsymbol{H w}_1 \|_2^2 > 0$  for every non-zero vector $\boldsymbol{w}$, hence $\boldsymbol{A}$ is  positive definite.

For every $\boldsymbol{\theta}_1 \in \boldsymbol{\Theta}$ and $\boldsymbol{\theta}_2  \in \boldsymbol{\Theta}$, take $\boldsymbol{z}^{\ast} = \mathrm{col}(\boldsymbol{\lambda}^{\ast},\boldsymbol{\nu}^{\ast}) $ and $\boldsymbol{z}^{\#} = \mathrm{col}(\boldsymbol{\lambda}^{\#},\boldsymbol{\nu}^{\#})$ be the optimal solutions of the optimization problem (\ref{zAz}) corresponding to $\boldsymbol{\theta}_1$ and $\boldsymbol{\theta}_2$, respectively. According to the KKT conditions, we have
            \begin{subequations}
              \begin{align} \label{43a}
                &\boldsymbol{Az}^{\ast} + \boldsymbol{b}^T\boldsymbol{F}_{\boldsymbol{\theta}_1}  -\frac{1}{2}\boldsymbol{B}^T\boldsymbol{\beta}_1 = 0 \\ \label{43b}
                &\boldsymbol{0} \leq \boldsymbol{Bz}^{\ast}  \perp \boldsymbol{\beta}_1 \geq \boldsymbol{0} , \boldsymbol{\beta}_1\in \mathcal{R}^{m+p}\\
                \label{43c}
                &\boldsymbol{Az}^{\#} + \boldsymbol{b}^T\boldsymbol{F}_{\boldsymbol{\theta}_2}   -\frac{1}{2}\boldsymbol{B}^T\boldsymbol{\beta}_2 = 0\\ \label{43d}
                &\boldsymbol{0} \leq \boldsymbol{Bz}^{\#}  \perp \boldsymbol{\beta}_2 \geq \boldsymbol{0}, \boldsymbol{\beta}_2\in \mathcal{R}^{m+p}.
              \end{align}
            \end{subequations}
Subtracting (\ref{43c}) from (\ref{43a}), we obtain
                \begin{align}\label{44}
                \boldsymbol{A}(\boldsymbol{z}^{\ast}-\boldsymbol{z}^{\#}) = -\boldsymbol{b}^T(\boldsymbol{F}_{\boldsymbol{\theta}_1}-\boldsymbol{F}_{\boldsymbol{\theta}_2}) + \frac{1}{2}\boldsymbol{B}^T(\boldsymbol{\beta}_1-\boldsymbol{\beta}_2).
                \end{align}
\noindent Multiplying the vector $(\boldsymbol{z}^{\ast}-\boldsymbol{z}^{\#})^T$ on both sides of (\ref{44}), we get
          \begin{equation} \label{45}
            \begin{aligned}
            &(\boldsymbol{z}^{\ast}-\boldsymbol{z}^{\#})^T\boldsymbol{A}(\boldsymbol{z}^{\ast}-\boldsymbol{z}^{\#})  \\
            = & - (\boldsymbol{z}^{\ast}-\boldsymbol{z}^{\#})^T\boldsymbol{b}^T(\boldsymbol{F}_{\boldsymbol{\theta}_1}-\boldsymbol{F}_{\boldsymbol{\theta}_2}) + (\boldsymbol{z}^{\ast}-\boldsymbol{z}^{\#})^T\boldsymbol{B}^T(\boldsymbol{\beta}_1-\boldsymbol{\beta}_2) \\
            =&- (\boldsymbol{z}^{\ast}-\boldsymbol{z}^{\#})^T\boldsymbol{b}^T(\boldsymbol{F}_{\boldsymbol{\theta}_1}-\boldsymbol{F}_{\boldsymbol{\theta}_2}) - \frac{1}{2}(\boldsymbol{Bz}^{\ast})^T\boldsymbol{\beta}_2 - \frac{1}{2}(\boldsymbol{Bz}^{\#})^T\boldsymbol{\beta}_1 \\
            \leq &- (\boldsymbol{z}^{\ast}-\boldsymbol{z}^{\#})^T\boldsymbol{b}^T(\boldsymbol{F}_{\boldsymbol{\theta}_1}-\boldsymbol{F}_{\boldsymbol{\theta}_2}),
            \end{aligned}
          \end{equation}
where the second equality holds because $(\boldsymbol{Bz}^{\ast})^T\boldsymbol{\beta}_1=0$ and $(\boldsymbol{Bz}^{\#})^T\boldsymbol{\beta}_2 = 0$, and the last inequality holds because $(\boldsymbol{Bz}^{\ast})^T\boldsymbol{\beta}_2 \geq 0$ and $(\boldsymbol{Bz}^{\#})^T\boldsymbol{\beta}_1 \geq 0$ from $\boldsymbol{Bz}^*$, $\boldsymbol{\beta}_1$, $\boldsymbol{Bz}^{\#}$ and $\boldsymbol{\beta}_2\geq \boldsymbol{0}$ in \eqref{43b} and \eqref{43d}.

Because $\boldsymbol{A}$ is a positive definite matrix, all its eigenvalues are positive. Then, we denote by the smallest eigenvalue $\lambda_{\min}(\boldsymbol{A}) > 0$. Therefore,
\begin{equation} \label{30q}
    \begin{aligned}
      \| \boldsymbol{z}^{\ast}-\boldsymbol{z}^{\#} \|_2^2 \leq &\frac{1}{\lambda_{\min}(\boldsymbol{A})}(\boldsymbol{z}^{\ast}-\boldsymbol{z}^{\#})^T\boldsymbol{A}(\boldsymbol{z}^{\ast}-\boldsymbol{z}^{\#})\\
    \overset{(\ref{45})}{\leq} & - \frac{1}{\lambda_{\min}(\boldsymbol{A})} (\boldsymbol{z}^{\ast}-\boldsymbol{z}^{\#})^T\boldsymbol{b}^T(\boldsymbol{F}_{\boldsymbol{\theta}_1}-\boldsymbol{F}_{\boldsymbol{\theta}_2}) \\
      \leq & \frac{1}{\lambda_{\min}(\boldsymbol{A})} \left\|\boldsymbol{z}^{\ast}-\boldsymbol{z}^{\#}\right\|_2\left\|\boldsymbol{b}\right\|_2\left\|\boldsymbol{F}_{\boldsymbol{\theta}_1}-\boldsymbol{F}_{\boldsymbol{\theta}_2}\right\|_2,
      \end{aligned}
    \end{equation}
where the last inequality  follows from the Cauchy-Schwartz inequality. Furthermore, we obtain from Proposition \ref{coro-1} that
\begin{equation} \label{31q}
    \begin{aligned}
       & \left\|\boldsymbol{z}^{\ast}-\boldsymbol{z}^{\#}\right\|_2
        \leq \frac{1}{\lambda_{\min}(\boldsymbol{A})} \| \boldsymbol{b} \|_2 \| \boldsymbol{F}_{\boldsymbol{\theta}_1}-\boldsymbol{F}_{\boldsymbol{\theta}_2} \|_2
       \\& \leq\frac{M_2}{\lambda_{\min}(\boldsymbol{A})} \| \boldsymbol{b} \|_2 \|\boldsymbol{\theta}_1 - \boldsymbol{\theta}_2 \|_2<\frac{2M_2B_2}{\lambda_{\min}(\boldsymbol{A})}  \|\boldsymbol{\theta}_1 - \boldsymbol{\theta}_2 \|_2,
      \end{aligned}
    \end{equation}
where the last inequality holds since $\| \boldsymbol{b} \|_2 \leq \|\nabla \boldsymbol{h}\|_2+\| \nabla \boldsymbol{g}\|_2<2B_2 $ from Assumption \ref{assu-2}(2). Thus, $\left\| \boldsymbol{\lambda}^{\ast} - \boldsymbol{\lambda}^{\#} \right\|_2 \leq \left\|\boldsymbol{z}^{\ast}-\boldsymbol{z}^{\#}\right\|_2  \leq  a \left\|\boldsymbol{\theta}_1 - \boldsymbol{\theta}_2 \right\|_2$ and $\left\| \boldsymbol{\nu}^{\ast} - \boldsymbol{\nu}^{\#} \right\|_2  \leq \left\|\boldsymbol{z}^{\ast}-\boldsymbol{z}^{\#}\right\|_2\leq  a \left\|\boldsymbol{\theta}_1 - \boldsymbol{\theta}_2 \right\|_2$ hold with $a\triangleq \frac{2M_2B_2}{\lambda_{\min}(\boldsymbol{A})}  $.

From (\ref{43a}) and (\ref{43b}), we have $\boldsymbol{z}^{\ast T}\boldsymbol{Az}^{\ast}  = - \boldsymbol{z}^{\ast T}\boldsymbol{b}^T\boldsymbol{F}_{\boldsymbol{\theta}_1}  + \frac{1}{2}\boldsymbol{z}^{\ast T}\boldsymbol{B}^T\boldsymbol{\beta}_1 = - \boldsymbol{z}^{\ast T} \boldsymbol{b}^T \boldsymbol{F}_{\boldsymbol{\theta}_1}$.
Similar to the derivation of (\ref{30q}) and (\ref{31q}), we have $\| \boldsymbol{z}^{\ast} \|_2 \leq \frac{1}{\lambda_{\min}(\boldsymbol{A})} \left\| \boldsymbol{b}^T \boldsymbol{F}_{\boldsymbol{\theta}_1} \right\|_2 \leq  \frac{2M_1B_2}{\lambda_{\min}(\boldsymbol{A})} $,
where the last inequality is obtained by  using $\| \boldsymbol{b} \|_2  <2B_2 $ and $\|\boldsymbol{F}_{\boldsymbol{\theta}}\|<M_1$ in Proposition \ref{coro-1}. And because $\| \boldsymbol{\lambda}^{\ast} \| _2\leq \| \boldsymbol{z}^{\ast} \|_2$ and $\| \boldsymbol{\nu}^{\ast}  \| _2\leq \| \boldsymbol{z}^{\ast} \|_2$, we reveal that $\| \boldsymbol{\lambda}^{\ast} \| _2,  \| \boldsymbol{\nu}^{\ast}  \| _2 \leq d$, where $d=\frac{2M_1B_2}{\lambda_{\min}(\boldsymbol{A})}  $.
Then the boundedness of $\boldsymbol{\lambda}^{\#}$ and $\boldsymbol{\nu}^{\#}$ can be obtained in the same way. So, the lemma holds.
\end{proof}

\emph{Proof of Proposition \ref{lema2}}:
For every $ \boldsymbol{u} \in \boldsymbol{U} ,  \boldsymbol{y}$ and every $\boldsymbol{\theta}_1, \boldsymbol{\theta}_2 \in  \boldsymbol{\Theta} $, let $\left(\boldsymbol{\lambda}^{\ast },\boldsymbol{\nu}^{\ast}\right)$ and $\left(\boldsymbol{\lambda}^{\# },\boldsymbol{\nu}^{\#}\right)$ be the optimal solutions of the optimization problem (\ref{op_pro}) corresponding to $\boldsymbol{\theta}_1$, $\boldsymbol{\theta}_2$, respectively. Then we have that  $l(\boldsymbol{\theta}_1;\boldsymbol{y},\boldsymbol{u}) = L(\boldsymbol{\theta}_1,\boldsymbol{\lambda}^*, \boldsymbol{\nu}^*;\boldsymbol{y},\boldsymbol{u})$ and $l(\boldsymbol{\theta}_2;\boldsymbol{y},\boldsymbol{u}) = L(\boldsymbol{\theta}_2,\boldsymbol{\lambda}^{\# }, \boldsymbol{\nu}^{\# };\boldsymbol{y},\boldsymbol{u})$, where $L(\cdot;\boldsymbol{y},\boldsymbol{u})$ is same as (\ref{newL}).

Without loss of generality, suppose $l(\boldsymbol{\theta}_1;\boldsymbol{y},\boldsymbol{u}) \geq l(\boldsymbol{\theta}_2;\boldsymbol{y},\boldsymbol{u})$. We obtain
\begin{equation} \label{l1l2}
	\begin{aligned}
		&|l(\boldsymbol{\theta}_1;\boldsymbol{y},\boldsymbol{u}) - l(\boldsymbol{\theta}_2;\boldsymbol{y},\boldsymbol{u}) | =  l(\boldsymbol{\theta}_1;\boldsymbol{y},\boldsymbol{u}) - l(\boldsymbol{\theta}_2;\boldsymbol{y},\boldsymbol{u}) \\
		= & \left\| \boldsymbol{F}_{\boldsymbol{\theta}_1}+\nabla \boldsymbol{h}\boldsymbol{\lambda}^{\ast} + \nabla \boldsymbol{g \nu}^{\ast} \right\|_2^2 - \| \boldsymbol{F}_{\boldsymbol{\theta}_2}+ \nabla \boldsymbol{h\lambda}^{\#} + \nabla \boldsymbol{g \nu}^{\#} \|_2^2 \\
		&+ \| \boldsymbol{H\lambda}^{\ast} \|_2^2 - \| \boldsymbol{H\lambda}^{\#} \|_2^2 \\
		= & \Bigl\langle \boldsymbol{F}_{\boldsymbol{\theta}_1}-\boldsymbol{F}_{\boldsymbol{\theta}_2} + \nabla \boldsymbol{h}(\boldsymbol{\lambda}^{\ast} - \boldsymbol{\lambda}^{\#}) + \nabla \boldsymbol{g} (\boldsymbol{\nu}^{\ast} - \boldsymbol{\nu}^{\#}), \\
	  & \ \ \boldsymbol{F}_{\boldsymbol{\theta}_1} +\nabla \boldsymbol{h}\boldsymbol{\lambda}^{\ast} + \nabla \boldsymbol{g \nu}^{\ast} + \boldsymbol{F}_{\boldsymbol{\theta}_2}+\nabla \boldsymbol{h\lambda}^{\#} + \nabla \boldsymbol{g \nu}^{\#} \Bigr\rangle \\
		& + \langle  \boldsymbol{H}(\boldsymbol{\lambda}^{\ast}-\boldsymbol{\lambda}^{\#}), \boldsymbol{H}(\boldsymbol{\lambda}^{\ast} + \boldsymbol{\lambda}^{\#}) \rangle \\
		 \leq & \left( \|\boldsymbol{F}_{\boldsymbol{\theta}_1}-\boldsymbol{F}_{\boldsymbol{\theta}_2}\|_2 + \| \nabla \boldsymbol{h}(\boldsymbol{\lambda}^{\ast} - \boldsymbol{\lambda}^{\#}) \|_2 + \| \nabla \boldsymbol{g} (\boldsymbol{\nu}^{\ast} - \boldsymbol{\nu}^{\#}) \|_2 \right)\\
		 &\cdot \left(\|\boldsymbol{F}_{\boldsymbol{\theta}_1} +\nabla \boldsymbol{h}\boldsymbol{\lambda}^{\ast} + \nabla \boldsymbol{g \nu}^{\ast} \|_2 + \|\boldsymbol{F}_{\boldsymbol{\theta}_2}+\nabla \boldsymbol{h\lambda}^{\#} + \nabla \boldsymbol{g \nu}^{\#} \|_2\right) \\
	   &+ \| \boldsymbol{H}(\boldsymbol{\lambda}^{\ast}-\boldsymbol{\lambda}^{\#})\|_2 \| \boldsymbol{H}(\boldsymbol{\lambda}^{\ast} + \boldsymbol{\lambda}^{\#})\|_2,
	  \end{aligned}
\end{equation}
where the last inequality holds by the Cauchy-Schwartz inequality. 
According to the boundedness of related variables from Assumption \ref{assu-2}, Proposition \ref{coro-1} and Lemma \ref{lem-1}, there exists a real number $ E >0$ such that $\| \boldsymbol{H}(\boldsymbol{\lambda}^{\ast} + \boldsymbol{\lambda}^{\#})\|_2, \| \boldsymbol{F}_{\boldsymbol{\theta}_1} +\nabla \boldsymbol{h}\boldsymbol{\lambda}^{\ast} + \nabla \boldsymbol{g \nu}^{\ast} \|_2, \| \boldsymbol{F}_{\boldsymbol{\theta}_2}+\nabla \boldsymbol{h\lambda}^{\#} + \nabla \boldsymbol{g \nu}^{\#} \|_2 \leq E$.
This together with (\ref{l1l2}) implies that
\begin{align*}
&|l(\boldsymbol{\theta}_1;\boldsymbol{y},\boldsymbol{u}) - l(\boldsymbol{\theta}_2;\boldsymbol{y},\boldsymbol{u}) |
\leq E \| \boldsymbol{H}(\boldsymbol{\lambda}^{\ast}-\boldsymbol{\lambda}^{\#})\|_2+\\& 2E\big(\|\boldsymbol{F}_{\boldsymbol{\theta}_1}-\boldsymbol{F}_{\boldsymbol{\theta}_2}\|_2 + \| \nabla \boldsymbol{h}(\boldsymbol{\lambda}^{\ast} - \boldsymbol{\lambda}^{\#}) \|_2 + \| \nabla \boldsymbol{g} (\boldsymbol{\nu}^{\ast} - \boldsymbol{\nu}^{\#}) \|_2  \big) .
\end{align*}
Then by recalling  the Lipschitz continuity of $\boldsymbol{F}_{\boldsymbol{\theta}}$ in \eqref{lip_F}, using Lemma \ref{lem-1} and Assumption \ref{assu-2}(2),   we have \[|l(\boldsymbol{\theta}_1;\boldsymbol{y},\boldsymbol{u}) - l(\boldsymbol{\theta}_2;\boldsymbol{y},\boldsymbol{u}) | \leq C \| \boldsymbol{\theta}_1 - \boldsymbol{\theta}_2 \| _2,\]
where $C=E(2M_2+5aB_2)$.
\hfill $\Box$

\section{Proof of Proposition \ref{conv}}
To prove Proposition \ref{conv}, we will give the following lemma.
\begin{lemma}
	\label{lem-3}
	Let Assumption \ref{assu-4} hold. We obtain that

(1) $\nabla \boldsymbol{g}^T\nabla \boldsymbol{g}$ is reversible.

(2) \begin{equation}
			\begin{aligned} \label{rdef}
				\boldsymbol{R} :=  \boldsymbol{I}-\nabla \boldsymbol{g}(\nabla \boldsymbol{g}^T\nabla \boldsymbol{g})^{-1}\nabla \boldsymbol{g}^T
			\end{aligned}
		\end{equation}
\noindent is  positive semidefinite and moreover, both
 $\nabla \boldsymbol{h}^T\boldsymbol{R}\nabla \boldsymbol{h}  $ and $\nabla \boldsymbol{h}^T\boldsymbol{R}\nabla \boldsymbol{h} + \boldsymbol{H}^2$ are reversible.

(3) \begin{equation}
			\begin{aligned} \label{qdef}
				\boldsymbol{Q} :=  \boldsymbol{R}-\boldsymbol{R}\nabla \boldsymbol{h}(\nabla \boldsymbol{h}^T\boldsymbol{R}\nabla \boldsymbol{h} + \boldsymbol{H}^2)^{-1}\nabla \boldsymbol{h}^T \boldsymbol{R}
			\end{aligned}
		\end{equation}
\noindent is a positive semidefinite matrix.
\end{lemma}
\begin{proof}
(1) Because $\nabla \boldsymbol{g}$ is a column full rank matrix, by the Singular Value Decomposition (SVD), we have
    \begin{align} \label{411}
          \nabla \boldsymbol{g} = \boldsymbol{U} \begin{bmatrix}
            \boldsymbol{\Sigma} _p\\
           \boldsymbol{0}
           \end{bmatrix}\boldsymbol{V}^T,
    \end{align}
where $\boldsymbol{U}$ and $\boldsymbol{V}$ are orthogonal matrices, and $\boldsymbol{\Sigma} _p$ is a $p$ dimensional diagonal matrix with non-zero diagonal elements.
Thus, $\nabla \boldsymbol{g}^T\nabla \boldsymbol{g} = \boldsymbol{V\Sigma}_p^2\boldsymbol{V}^T$ is reversible and positive definite.

(2) Putting  (\ref{411}) into $\boldsymbol{R}$, we have
	\begin{equation} \label{412}
		\begin{aligned}
			\boldsymbol{R} & =\boldsymbol{I}-\nabla \boldsymbol{g}(\nabla \boldsymbol{g}^T\nabla \boldsymbol{g})^{-1}\nabla \boldsymbol{g}^T \\
			& = \boldsymbol{I}- \boldsymbol{U}\begin{bmatrix}
				\boldsymbol{\Sigma} _p\\
				\boldsymbol{0}
			 \end{bmatrix}\boldsymbol{V}^T(\boldsymbol{V}\boldsymbol{\Sigma}_p^2\boldsymbol{V}^T)^{-1}\boldsymbol{V}\begin{bmatrix}
			  \boldsymbol{\Sigma} _p & \boldsymbol{0}
			 \end{bmatrix}\boldsymbol{U}^T \\
			 & = \boldsymbol{I}- \boldsymbol{U}\begin{bmatrix}
				\boldsymbol{\Sigma} _p\\
				\boldsymbol{0}
			 \end{bmatrix}\boldsymbol{V}^T\boldsymbol{V}\boldsymbol{\Sigma}_p^{-2}\boldsymbol{V}^T\boldsymbol{V}\begin{bmatrix}
				\boldsymbol{\Sigma} _p & \boldsymbol{0}
			 \end{bmatrix}U^T \\
			 & = \boldsymbol{I}- \boldsymbol{U}\begin{bmatrix}
				\boldsymbol{\Sigma} _p\\
				\boldsymbol{0}
			 \end{bmatrix}\boldsymbol{\Sigma}_p^{-2}\begin{bmatrix}
				\boldsymbol{\Sigma} _p & \boldsymbol{0}
			 \end{bmatrix}\boldsymbol{U}^T \\&
=\boldsymbol{I}- \boldsymbol{U}\begin{bmatrix}
			  \boldsymbol{I}_p & \boldsymbol{0}\\
			  \boldsymbol{0} & \boldsymbol{0}
			 \end{bmatrix}\boldsymbol{U}^T
 = \boldsymbol{U}\begin{bmatrix}
			  \boldsymbol{0} & \boldsymbol{0}\\
			  \boldsymbol{0} & \boldsymbol{I}_{n-p}
			 \end{bmatrix}\boldsymbol{U}^T.
		\end{aligned}
	\end{equation}
Thus, $\boldsymbol{R}$ is a positive semidefinite matrix. Note that $\boldsymbol{R}^2 = \boldsymbol{U}\begin{bmatrix}
    \boldsymbol{0} & \boldsymbol{0}\\
    \boldsymbol{0} & \boldsymbol{I}_{n-p}
     \end{bmatrix}\boldsymbol{U}^T\boldsymbol{U}\begin{bmatrix}
    \boldsymbol{0} & \boldsymbol{0}\\
    \boldsymbol{0} & \boldsymbol{I}_{n-p}
     \end{bmatrix}\boldsymbol{U}^T=\boldsymbol{R}$.

In the following, we will show that $\boldsymbol{R}\nabla \boldsymbol{h}$ is a column full rank matrix by contradiction. Suppose that there exists a non-zero vector $\boldsymbol{x} \in \mathcal{R}^{m}$ such that $\boldsymbol{R}\nabla \boldsymbol{hx} = 0$, i.e., $\boldsymbol{R}\nabla \boldsymbol{hx}=\boldsymbol{U}\begin{bmatrix}
	\boldsymbol{0} & \boldsymbol{0}\\
	\boldsymbol{0} & \boldsymbol{I}_{n-p}
 \end{bmatrix}\boldsymbol{U}^T\nabla \boldsymbol{hx} = 0$.
Since $\boldsymbol{U} \in \mathcal{R}^{n \times n}$ is a full rank square matrix, we have
      \begin{align} \label{413}
        \begin{bmatrix}
          \boldsymbol{0} & \boldsymbol{0}\\
          \boldsymbol{0} & \boldsymbol{I}_{n-p}
         \end{bmatrix}\boldsymbol{U}^T\nabla \boldsymbol{hx} = \boldsymbol{0}.
      \end{align}
      Note that
      \begin{align}
          \begin{bmatrix}
             \boldsymbol{U}^T\nabla \boldsymbol{h} & \begin{bmatrix}
                 \boldsymbol{\Sigma} _p\\
                \boldsymbol{0}
                \end{bmatrix}
          \end{bmatrix}  =& \boldsymbol{U}^T\begin{bmatrix}
           \nabla \boldsymbol{h} & \boldsymbol{U} \begin{bmatrix}
             \boldsymbol{\Sigma} _p\\
            \boldsymbol{0}
            \end{bmatrix}\boldsymbol{V}^T
          \end{bmatrix}\begin{bmatrix}
           \boldsymbol{I}_p & \boldsymbol{0} \\
           \boldsymbol{0}  & \boldsymbol{V}
          \end{bmatrix}  \notag \\
          =& \boldsymbol{U}^T\begin{bmatrix}
           \nabla \boldsymbol{h} & \nabla \boldsymbol{g}
          \end{bmatrix}\begin{bmatrix}
           \boldsymbol{I}_p & \boldsymbol{0} \\
           \boldsymbol{0}  & \boldsymbol{V}
          \end{bmatrix}.
       \end{align}
Then under Assumption \ref{assu-4}, we conclude that
\[\mathrm{rank}\left(\begin{bmatrix}
	\boldsymbol{U}^T\nabla \boldsymbol{h} & \begin{bmatrix}
	  \boldsymbol{\Sigma} _p\\
	 \boldsymbol{0}
	 \end{bmatrix}
   \end{bmatrix}\right)
   = \mathrm{rank}\left(\begin{bmatrix}
	\nabla \boldsymbol{h} & \nabla \boldsymbol{g}
   \end{bmatrix}\right)
   =m + p.\]
Therefore,
  $\begin{bmatrix}
      \boldsymbol{U}^T\nabla \boldsymbol{h} & \begin{bmatrix}
          \boldsymbol{\Sigma} _p\\
         \boldsymbol{0}
         \end{bmatrix}
     \end{bmatrix}$ is a column full rank matrix. Furthermore, we have $\begin{bmatrix}
	\boldsymbol{U}^T\nabla \boldsymbol{h} & \begin{bmatrix}
	  \boldsymbol{\Sigma} _p\\
	 \boldsymbol{0}
	 \end{bmatrix}
   \end{bmatrix}\begin{bmatrix}
	\boldsymbol{x} \\
	-\begin{bmatrix}
	  \boldsymbol{\Sigma}_p^{-1} & \boldsymbol{0}
	 \end{bmatrix} \boldsymbol{U}^T\nabla \boldsymbol{hx}
  \end{bmatrix} \\
  = \boldsymbol{U}^T\nabla \boldsymbol{hx} - \begin{bmatrix}
	\boldsymbol{\Sigma} _p\\
   \boldsymbol{0}
   \end{bmatrix} \begin{bmatrix}
	\boldsymbol{\Sigma}_p^{-1} & \boldsymbol{0}
   \end{bmatrix} \boldsymbol{U}^T\nabla \boldsymbol{hx}=\begin{bmatrix}
	\boldsymbol{0} & \boldsymbol{0}\\
	\boldsymbol{0} & \boldsymbol{I}_{n-p}
   \end{bmatrix}\boldsymbol{U}^T\nabla \boldsymbol{hx}
   \overset{(\ref{413})}{=}  \boldsymbol{0}$.
Since $\begin{bmatrix}
      \boldsymbol{U}^T\nabla \boldsymbol{h} & \begin{bmatrix}
          \boldsymbol{\Sigma} _p\\
         \boldsymbol{0}
         \end{bmatrix}
     \end{bmatrix}$ is a column full rank matrix, we get $\begin{bmatrix}
        \boldsymbol{x} \\
        -\begin{bmatrix}
          \boldsymbol{\Sigma}_p^{-1} & \boldsymbol{0}
         \end{bmatrix} \boldsymbol{U}^T\nabla \boldsymbol{hx}
      \end{bmatrix} = \boldsymbol{0}$.
This contradicts that $\boldsymbol{x}$ is a non-zero vector, so $\boldsymbol{R}\nabla \boldsymbol{h}$ is a column full rank matrix.

Similar to the method of proving that $\nabla \boldsymbol{g}^T\nabla \boldsymbol{g}$ is reversible and positive definite, $\nabla \boldsymbol{h}^T\boldsymbol{RR}\nabla \boldsymbol{h}$ is reversible and positive definite. Thus, $\nabla \boldsymbol{h}^T\boldsymbol{R}\nabla \boldsymbol{h}$ is reversible and positive definite because $\boldsymbol{R}^2=\boldsymbol{R}$. Furthermore, $\nabla \boldsymbol{h}^T\boldsymbol{R}\nabla \boldsymbol{h} + \boldsymbol{H}^2$ is also reversible and positive definite.

(3) Denote $\boldsymbol{Q}_1 :=  \boldsymbol{R}-\boldsymbol{R}\nabla \boldsymbol{h}(\nabla \boldsymbol{h}^T\boldsymbol{R}\nabla \boldsymbol{h})^{-1}\nabla \boldsymbol{h}^T \boldsymbol{R}$. Because $\boldsymbol{R}^2=\boldsymbol{R}$, we have
   \begin{align}\label{def-Q1}
     \boldsymbol{Q}_1 &= \boldsymbol{R}^2-\boldsymbol{R}^2\nabla \boldsymbol{h}(\nabla \boldsymbol{h}^T\boldsymbol{R}^2\nabla \boldsymbol{h})^{-1}\nabla \boldsymbol{h}^T \boldsymbol{R}^2 \\
     &= \boldsymbol{R} \left\{\boldsymbol{I} - (\boldsymbol{R}\nabla \boldsymbol{h})\left [(\boldsymbol{R}\nabla \boldsymbol{h})^T(\boldsymbol{R}\nabla \boldsymbol{h})\right ]^{-1}(\boldsymbol{R}\nabla \boldsymbol{h})^T\right \} \boldsymbol{R}. \notag
   \end{align}
Since  $\boldsymbol{R}\nabla \boldsymbol{h}$ is  column full rank,
similarly to the method of proving that $\boldsymbol{R}$ is positive semidefinite in (\ref{412}), we can also show that $\boldsymbol{I} - (\boldsymbol{R}\nabla \boldsymbol{h})\left [(\boldsymbol{R}\nabla \boldsymbol{h})^T(\boldsymbol{R}\nabla \boldsymbol{h})\right ]^{-1}(\boldsymbol{R}\nabla \boldsymbol{h})^T$ is positive semidefinite. Thus, $\boldsymbol{Q}_1$ is positive semidefinite.

In the following, we will prove that $\boldsymbol{Q}$ is positive semidefinite by showing that $\boldsymbol{y}^T(\boldsymbol{Q} - \boldsymbol{Q}_1)\boldsymbol{y} \geq 0$ for every $\boldsymbol{y} \in \mathcal{R}^n$.
According to Woodbury Matrix Inversion Formula $\left(\boldsymbol{Z}+\boldsymbol{XY}\right)^{-1} = \boldsymbol{Z}^{-1}-\boldsymbol{Z}^{-1}\boldsymbol{X}(\boldsymbol{I}+\boldsymbol{YZ}^{-1}\boldsymbol{X})^{-1}\boldsymbol{YZ}^{-1}$,
with $\boldsymbol{Z}=\nabla \boldsymbol{h}^T\boldsymbol{R}\nabla \boldsymbol{h}$ and $\boldsymbol{X}=\boldsymbol{Y}=\boldsymbol{H}$, we have $\left(\nabla \boldsymbol{h}^T\boldsymbol{R}\nabla \boldsymbol{h} + \boldsymbol{H}^2\right)^{-1} = \left(\nabla \boldsymbol{h}^T\boldsymbol{R}\nabla \boldsymbol{h}\right)^{-1} - \left(\nabla \boldsymbol{h}^T\boldsymbol{R}\nabla \boldsymbol{h}\right)^{-1}\boldsymbol{H}\left(\boldsymbol{I}+\boldsymbol{H}(\nabla \boldsymbol{h}^T\boldsymbol{R}\nabla \boldsymbol{h})^{-1}\boldsymbol{H}\right)^{-1}\boldsymbol{H}\left(\nabla \boldsymbol{h}^T\boldsymbol{R}\nabla \boldsymbol{h}\right)^{-1}$.
This, together with \eqref{def-Q1} and \eqref{qdef}, implies that
\begin{align*}
	& \boldsymbol{Q} - \boldsymbol{Q}_1 = \left[\boldsymbol{H}(\nabla \boldsymbol{h}^T\boldsymbol{R}\nabla \boldsymbol{h})^{-1}\nabla \boldsymbol{h}^T\boldsymbol{R}\right]^T \\
  & \cdot \left(\boldsymbol{I}+\boldsymbol{H}(\nabla \boldsymbol{h}^T\boldsymbol{R}\nabla \boldsymbol{h})^{-1}\boldsymbol{H}\right)^{-1}\left[\boldsymbol{H}(\nabla \boldsymbol{h}^T\boldsymbol{R}\nabla \boldsymbol{h})^{-1}\nabla \boldsymbol{h}^T\boldsymbol{R}\right].
\end{align*}
Because $(\nabla \boldsymbol{h}^T\boldsymbol{R}\nabla \boldsymbol{h})^{-1} = \boldsymbol{V}^{'}\boldsymbol{\Sigma}_p^{'-2}\boldsymbol{V}^{'T}$ is the positive definite matrix from SVD ($\boldsymbol{V}^{'}$ is an orthogonal square matrix, and $\boldsymbol{\Sigma}_p^{'-2}$ is a diagonal matrix with positive diagonal elements), we get that $\left(\boldsymbol{I}+\boldsymbol{H}(\nabla \boldsymbol{h}^T\boldsymbol{R}\nabla \boldsymbol{h})^{-1}\boldsymbol{H}\right)^{-1}$ is positive definite.
Then for every $\boldsymbol{y} \in \mathcal{R}^n$, we obtain that
\begin{align*}
	&\boldsymbol{y}^T(\boldsymbol{Q} - \boldsymbol{Q}_1)\boldsymbol{y} = \left[\boldsymbol{H}(\nabla \boldsymbol{h}^T\boldsymbol{R}\nabla \boldsymbol{h})^{-1}\nabla \boldsymbol{h}^T\boldsymbol{Ry}\right]^T \\
  & \cdot \left(\boldsymbol{I}+\boldsymbol{H}(\nabla \boldsymbol{h}^T\boldsymbol{R}\nabla \boldsymbol{h})^{-1}\boldsymbol{H}\right)^{-1}\left[\boldsymbol{H}(\nabla \boldsymbol{h}^T\boldsymbol{R}\nabla \boldsymbol{h})^{-1}\nabla \boldsymbol{h}^T\boldsymbol{Ry}\right] \geq 0.
\end{align*}
Therefore, $\boldsymbol{Q}$ is a positive semidefinite matrix.
\end{proof}

\emph{Proof of Proposition \ref{conv}}:
Define  $\boldsymbol{k} := \left(\nabla \boldsymbol{g}^T\nabla \boldsymbol{g}\right)^{-1}\nabla \boldsymbol{g}^T(\boldsymbol{F}_{\boldsymbol{\theta}} + \nabla \boldsymbol{h\lambda})$,
where   $\nabla \boldsymbol{g}^T\nabla \boldsymbol{g}$ is reversible from Lemma \ref{lem-3}(1).
Then   $L\left(\boldsymbol{\theta}, \boldsymbol{\lambda}, \boldsymbol{\nu};\boldsymbol{y}, \boldsymbol{u}\right)$ can be rewritten as
\begin{equation} \label{42qw0}
	\begin{aligned}
   & \left\|\boldsymbol{F}_{\boldsymbol{\theta}} + \nabla \boldsymbol{h}\boldsymbol{\lambda} + \nabla \boldsymbol{g} \boldsymbol{\nu}\right\|_2^2 + \left\|\boldsymbol{H}\boldsymbol{\lambda}\right\|_2^2 + \|\boldsymbol{g}\|_2^2 \\
  =&\boldsymbol{\nu}^T\nabla \boldsymbol{g}^T\nabla \boldsymbol{g} \boldsymbol{\nu}+2(\boldsymbol{F}_{\boldsymbol{\theta}} + \nabla \boldsymbol{h}\boldsymbol{\lambda} )^T \nabla \boldsymbol{g} \boldsymbol{\nu}
 \\&+\boldsymbol{F}_{\boldsymbol{\theta}}^T\boldsymbol{F}_{\boldsymbol{\theta}}+  2\boldsymbol{\lambda}^T\nabla\boldsymbol{h}^T\boldsymbol{F}_{\boldsymbol{\theta}} +\boldsymbol{\lambda}^T \nabla\boldsymbol{h}^T\nabla\boldsymbol{h} \boldsymbol{\lambda} + \left\|\boldsymbol{H}\boldsymbol{\lambda}\right\|_2^2 + \|\boldsymbol{g}\|_2^2
  \\= & (\boldsymbol{\nu} + \boldsymbol{k})^T\nabla \boldsymbol{g}^T\nabla \boldsymbol{g}(\boldsymbol{\nu} + \boldsymbol{k}) - \boldsymbol{k}^T\nabla\boldsymbol{g}^T\nabla\boldsymbol{g}\boldsymbol{k}
 \\&+  \boldsymbol{\lambda}^T(\boldsymbol{H}^2+ \nabla\boldsymbol{h}^T\nabla\boldsymbol{h} ) \boldsymbol{\lambda} + 2\boldsymbol{\lambda}^T\nabla\boldsymbol{h}^T\boldsymbol{F}_{\boldsymbol{\theta}} + \boldsymbol{F}_{\boldsymbol{\theta}}^T\boldsymbol{F}_{\boldsymbol{\theta}} + \boldsymbol{g}^T\boldsymbol{g}.
	\end{aligned}
\end{equation}
Note that
\begin{align}\label{kgk}
-&\boldsymbol{k}^T\nabla\boldsymbol{g}^T\nabla\boldsymbol{g}\boldsymbol{k}
=-(\boldsymbol{F}_{\boldsymbol{\theta}} + \nabla \boldsymbol{h\lambda})^T \nabla \boldsymbol{g}(\nabla \boldsymbol{g}^T\nabla \boldsymbol{g})^{-1}\nabla \boldsymbol{g}^T  (\boldsymbol{F}_{\boldsymbol{\theta}} + \nabla \boldsymbol{h\lambda}) \notag
\\& \overset{\eqref{rdef}}{=}(\boldsymbol{F}_{\boldsymbol{\theta}} + \nabla \boldsymbol{h\lambda})^T ( \boldsymbol{R}-\boldsymbol{I}) (\boldsymbol{F}_{\boldsymbol{\theta}} + \nabla \boldsymbol{h\lambda})
\\&=\boldsymbol{F}_{\boldsymbol{\theta}}^T ( \boldsymbol{R}-\boldsymbol{I})  \boldsymbol{F}_{\boldsymbol{\theta}}+2\boldsymbol{\lambda}^T\nabla\boldsymbol{h}^T( \boldsymbol{R}-\boldsymbol{I}) \boldsymbol{F}_{\boldsymbol{\theta}} +\boldsymbol{\lambda}^T \nabla\boldsymbol{h}^T( \boldsymbol{R}-\boldsymbol{I}) \nabla\boldsymbol{h} \boldsymbol{\lambda}  .\notag
\end{align}
Define $\boldsymbol{P} := \nabla \boldsymbol{h}^T\boldsymbol{R}\nabla \boldsymbol{h} + \boldsymbol{H}^2$, which is reversible from Lemma \ref{lem-3}(2).
Then by substituting  \eqref{kgk} into   \eqref{42qw0}, we obtain that
\begin{equation} \label{42qw}
	\begin{aligned}
   & L\left(\boldsymbol{\theta}, \boldsymbol{\lambda}, \boldsymbol{\nu};\boldsymbol{y}, \boldsymbol{u}\right)
  =  (\boldsymbol{\nu} + \boldsymbol{k})^T\nabla \boldsymbol{g}^T\nabla \boldsymbol{g}(\boldsymbol{\nu} + \boldsymbol{k}) + \boldsymbol{g}^T\boldsymbol{g}
 \\&+  \boldsymbol{\lambda}^T\boldsymbol{P} \boldsymbol{\lambda} + 2\boldsymbol{\lambda}^T\nabla\boldsymbol{h}^T \boldsymbol{R}\boldsymbol{F}_{\boldsymbol{\theta}} + \boldsymbol{F}_{\boldsymbol{\theta}}^T \boldsymbol{R} \boldsymbol{F}_{\boldsymbol{\theta}}\\
	\overset{\eqref{qdef}}{=} & \left[ \boldsymbol{\lambda} + \boldsymbol{P}^{-1}\nabla \boldsymbol{h}^T\boldsymbol{RF}_{\boldsymbol{\theta}}\right]^T\boldsymbol{P} \left[ \boldsymbol{\lambda} + \boldsymbol{P}^{-1}\nabla \boldsymbol{h}^T\boldsymbol{RF}_{\boldsymbol{\theta}} \right] \\
		& + (\boldsymbol{\nu} + \boldsymbol{k})^T\nabla \boldsymbol{g}^T\nabla \boldsymbol{g}(\boldsymbol{\nu} + \boldsymbol{k}) + \boldsymbol{F}_{\boldsymbol{\theta}}^T\boldsymbol{QF}_{\boldsymbol{\theta}} + \boldsymbol{g}^T\boldsymbol{g}.
	\end{aligned}
\end{equation}
Thus, we have
\begin{equation*}
	\begin{aligned}
		&l(\boldsymbol{\theta};\boldsymbol{y},\boldsymbol{u}) = \min_{\boldsymbol{\lambda} \geq \boldsymbol{0}, \boldsymbol{\nu}} L\left(\boldsymbol{\theta}, \boldsymbol{\lambda}, \boldsymbol{\nu};\boldsymbol{y}, \boldsymbol{u}\right) \\
		\overset{(\ref{42qw})}{=} &\min_{\boldsymbol{\lambda} \geq \boldsymbol{0}, \boldsymbol{\nu}} \left[ \boldsymbol{\lambda} + \boldsymbol{P}^{-1}\nabla \boldsymbol{h}^T\boldsymbol{RF}_{\boldsymbol{\theta}}\right]^T\boldsymbol{P} \left[ \boldsymbol{\lambda} + \boldsymbol{P}^{-1}\nabla \boldsymbol{h}^T\boldsymbol{RF}_{\boldsymbol{\theta}} \right] \\
		& + (\boldsymbol{\nu} + \boldsymbol{k})^T\nabla \boldsymbol{g}^T\nabla \boldsymbol{g}(\boldsymbol{\nu} + \boldsymbol{k}) + \boldsymbol{F}_{\boldsymbol{\theta}}^T\boldsymbol{QF}_{\boldsymbol{\theta}} + \boldsymbol{g}^T\boldsymbol{g} \\
		\overset{\boldsymbol{v}=-\boldsymbol{k}}{=} & \min_{\boldsymbol{\lambda} \geq \boldsymbol{0}} \left[ \boldsymbol{\lambda} + \boldsymbol{P}^{-1}\nabla \boldsymbol{h}^T\boldsymbol{RF}_{\boldsymbol{\theta}}\right]^T\boldsymbol{P} \left[ \boldsymbol{\lambda} + \boldsymbol{P}^{-1}\nabla \boldsymbol{h}^T\boldsymbol{RF}_{\boldsymbol{\theta}} \right]
\\& + \boldsymbol{F}_{\boldsymbol{\theta}}^T\boldsymbol{QF}_{\boldsymbol{\theta}} + \boldsymbol{g}^T\boldsymbol{g}.
	\end{aligned}
\end{equation*}

For every $\theta_1, \theta_2 \in  \Theta$ and every $\alpha,\beta \geq 0 ,\alpha + \beta = 1$, let $\boldsymbol{\lambda}_{\alpha \boldsymbol{\theta}_1 + \beta \boldsymbol{\theta}_2}^{\ast}$, $\boldsymbol{\lambda}_{\boldsymbol{\theta}_1}^{\ast}$ and $\boldsymbol{\lambda}_{\boldsymbol{\theta}_2}^{\ast}$ be the optimal solutions of the optimization problem (\ref{op_pro}) corresponding to $\alpha \boldsymbol{\theta}_1 + \beta \boldsymbol{\theta}_2, \boldsymbol{\theta}_1$ and $\boldsymbol{\theta}_2$, respectively. we have
	\begin{equation}\label{wer}
		\begin{aligned}
			&l(\alpha \boldsymbol{\theta}_1 + \beta \boldsymbol{\theta}_2;\boldsymbol{y},\boldsymbol{u}) =  \left[ \boldsymbol{\lambda}_{\alpha \boldsymbol{\theta}_1 + \beta \boldsymbol{\theta}_2}^{\ast} + \boldsymbol{P}^{-1}\nabla \boldsymbol{h}^T\boldsymbol{RF}_{\alpha \boldsymbol{\theta}_1 + \beta \boldsymbol{\theta}_2}\right]^T \\
			&\cdot \boldsymbol{P} \left[\boldsymbol{\lambda}_{\alpha \boldsymbol{\theta}_1 + \beta \boldsymbol{\theta}_2}^{\ast} + \boldsymbol{P}^{-1}\nabla \boldsymbol{h}^T\boldsymbol{RF}_{\alpha \boldsymbol{\theta}_1 + \beta \boldsymbol{\theta}_2}\right] + \boldsymbol{F}_{\alpha \boldsymbol{\theta}_1 + \beta \boldsymbol{\theta}_2}^T\boldsymbol{QF}_{\alpha \boldsymbol{\theta}_1 + \beta \boldsymbol{\theta}_2} + \boldsymbol{g}^T\boldsymbol{g}\\
			\leq & \left[ \alpha \boldsymbol{\lambda}_{\boldsymbol{\theta}_1}^{\ast} + \beta \boldsymbol{\lambda}_{\boldsymbol{\theta}_2}^{\ast} +\boldsymbol{P}^{-1}\nabla \boldsymbol{h}^T\boldsymbol{RF}_{\alpha \boldsymbol{\theta}_1 + \beta \boldsymbol{\theta}_2}\right]^T\\
			&\cdot \boldsymbol{P}\left[\alpha \boldsymbol{\lambda}_{\boldsymbol{\theta}_1}^{\ast} + \beta \boldsymbol{\lambda}_{\boldsymbol{\theta}_2}^{\ast} +\boldsymbol{P}^{-1}\nabla \boldsymbol{h}^T\boldsymbol{RF}_{\alpha \boldsymbol{\theta}_1 + \beta \boldsymbol{\theta}_2} \right] + \boldsymbol{F}_{\alpha \boldsymbol{\theta}_1 + \beta \boldsymbol{\theta}_2}^T\boldsymbol{QF}_{\alpha \boldsymbol{\theta}_1 + \beta \boldsymbol{\theta}_2} + \boldsymbol{g}^T\boldsymbol{g} \\
			=& \left[ \alpha(\boldsymbol{\lambda}_{\boldsymbol{\theta}_1}^{\ast}+\boldsymbol{P}^{-1}\nabla \boldsymbol{h}^T\boldsymbol{RF}_{\boldsymbol{\theta}_1})  + \beta (\boldsymbol{\lambda}_{\boldsymbol{\theta}_2}^{\ast}+\boldsymbol{P}^{-1}\nabla \boldsymbol{h}^T\boldsymbol{RF}_{\boldsymbol{\theta}_2})\right]^T\\
       \cdot & \boldsymbol{P}\left[\alpha(\boldsymbol{\lambda}_{\boldsymbol{\theta}_1}^{\ast}+\boldsymbol{P}^{-1}\nabla \boldsymbol{h}^T\boldsymbol{RF}_{\boldsymbol{\theta}_1})  + \beta (\boldsymbol{\lambda}_{\boldsymbol{\theta}_2}^{\ast}+\boldsymbol{P}^{-1}\nabla \boldsymbol{h}^T\boldsymbol{RF}_{\boldsymbol{\theta}_2}) \right] \\
	  &+ (\alpha \boldsymbol{F}_{\boldsymbol{\theta}_1}+ \beta \boldsymbol{F}_{\boldsymbol{\theta}_2})^T\boldsymbol{Q}(\alpha \boldsymbol{F}_{\boldsymbol{\theta}_1}+ \beta \boldsymbol{F}_{\boldsymbol{\theta}_2}) + \boldsymbol{g}^T\boldsymbol{g}.
		  \end{aligned}
	\end{equation}
where the first inequality holds	because $\boldsymbol{\lambda}_{\alpha \boldsymbol{\theta}_1 + \beta \boldsymbol{\theta}_2}^{\ast}$ is the optimal solution that minimizes the loss function, and $\boldsymbol{\lambda} = \alpha \boldsymbol{\lambda}_{\boldsymbol{\theta}_1}^{\ast} + \beta \boldsymbol{\lambda}_{\boldsymbol{\theta}_2}^{\ast}$ makes the function value no smaller;   and the last equality holds because $\boldsymbol{F}_{\alpha \boldsymbol{\theta}_1 + \beta \boldsymbol{\theta}_2}=\alpha \boldsymbol{F}_{\boldsymbol{\theta}_1} + \beta \boldsymbol{F}_{\boldsymbol{\theta}_2}$ by \eqref{F_theta}.
Therefore, for every $ \boldsymbol{\theta}_1, \boldsymbol{\theta}_2 \in  \boldsymbol{\Theta}$ and every $\alpha,\beta \geq 0 ,\alpha + \beta = 1$, we obtain
\begin{align*}
	&\alpha l(\boldsymbol{\theta}_1;\boldsymbol{y},\boldsymbol{u}) + \beta l(\boldsymbol{\theta}_2;\boldsymbol{y},\boldsymbol{u}) - l(\alpha \boldsymbol{\theta}_1 + \beta \boldsymbol{\theta}_2;\boldsymbol{y},\boldsymbol{u}) \\
\overset{(\ref{wer})}{\geq} & \alpha (\boldsymbol{\lambda}_{\boldsymbol{\theta}_1}^{\ast}+\boldsymbol{P}^{-1}\nabla \boldsymbol{h}^T\boldsymbol{RF}_{\boldsymbol{\theta}_1})^T\boldsymbol{P}(\boldsymbol{\lambda}_{\boldsymbol{\theta}_1}^{\ast}+\boldsymbol{P}^{-1}\nabla \boldsymbol{h}^T\boldsymbol{RF}_{\boldsymbol{\theta}_1}) \\
  +& \beta (\boldsymbol{\lambda}_{\boldsymbol{\theta}_2}^{\ast}+\boldsymbol{P}^{-1}\nabla \boldsymbol{h}^T\boldsymbol{RF}_{\boldsymbol{\theta}_2})^T\boldsymbol{P}(\boldsymbol{\lambda}_{\boldsymbol{\theta}_2}^{\ast}+\boldsymbol{P}^{-1}\nabla \boldsymbol{h}^T\boldsymbol{RF}_{\boldsymbol{\theta}_2})\\
	-&\Big[ \alpha(\boldsymbol{\lambda}_{\boldsymbol{\theta}_1}^{\ast}+\boldsymbol{P}^{-1}\nabla \boldsymbol{h}^T\boldsymbol{RF}_{\boldsymbol{\theta}_1}) + \beta (\boldsymbol{\lambda}_{\boldsymbol{\theta}_2}^{\ast}+\boldsymbol{P}^{-1}\nabla \boldsymbol{h}^T\boldsymbol{RF}_{\boldsymbol{\theta}_2})\Big]^T\\ \cdot&\boldsymbol{P}\left[\alpha(\boldsymbol{\lambda}_{\boldsymbol{\theta}_1}^{\ast}+\boldsymbol{P}^{-1}\nabla \boldsymbol{h}^T\boldsymbol{RF}_{\boldsymbol{\theta}_1}) + \beta (\boldsymbol{\lambda}_{\boldsymbol{\theta}_2}^{\ast}+\boldsymbol{P}^{-1}\nabla \boldsymbol{h}^T\boldsymbol{RF}_{\boldsymbol{\theta}_2}) \right] \\
	&+ \alpha \boldsymbol{F}_{\boldsymbol{\theta} _1}^T\boldsymbol{QF}_{\boldsymbol{\theta} _1} + \beta \boldsymbol{F}_{\boldsymbol{\theta} _2}^T\boldsymbol{QF}_{\boldsymbol{\theta} _2} - (\alpha \boldsymbol{F}_{\boldsymbol{\theta}_1}+ \beta \boldsymbol{F}_{\boldsymbol{\theta}_2})^T\boldsymbol{Q}(\alpha \boldsymbol{F}_{\boldsymbol{\theta}_1}+ \beta \boldsymbol{F}_{\boldsymbol{\theta}_2}) \\= &\alpha \beta (\boldsymbol{\lambda}_{\boldsymbol{\theta}_1}^{\ast}+\boldsymbol{P}^{-1}\nabla \boldsymbol{h}^T\boldsymbol{RF}_{\boldsymbol{\theta}_1} - \boldsymbol{\lambda}_{\boldsymbol{\theta}_2}^{\ast} - \boldsymbol{P}^{-1}\nabla \boldsymbol{h}^T\boldsymbol{RF}_{\boldsymbol{\theta}_2})^T\\
  &\cdot \boldsymbol{P} (\boldsymbol{\lambda}_{\boldsymbol{\theta}_1}^{\ast}+\boldsymbol{P}^{-1}\nabla \boldsymbol{h}^T\boldsymbol{RF}_{\boldsymbol{\theta}_1} - \boldsymbol{\lambda}_{\boldsymbol{\theta}_2}^{\ast} - \boldsymbol{P}^{-1}\nabla \boldsymbol{h}^T\boldsymbol{RF}_{\boldsymbol{\theta}_2}) \\
  &+ \alpha \beta (\boldsymbol{F}_{\boldsymbol{\theta}_1}-\boldsymbol{F}_{\boldsymbol{\theta}_2})^T\boldsymbol{Q}(\boldsymbol{F}_{\boldsymbol{\theta}_1}-\boldsymbol{F}_{\boldsymbol{\theta}_2}) \geq 0
\end{align*}
\noindent where the last inequality holds because $\boldsymbol{Q}$ and $\boldsymbol{P}$ are positive semidefinite matrices according to Lemma \ref{lem-3}.
	
Therefore, we conclude that the loss function is convex with respect to the parameter $\boldsymbol{\theta}$.
\hfill $\Box$
\end{document}